%% file: arxiv_version.tex
\documentclass[a4paper,UKenglish,cleveref, autoref,thm-restate]{lipics-v2021}

\pdfoutput=1 %
\hideLIPIcs  %

\title{Which graph motif parameters count?}

\author{Markus {Bläser}}
{Saarland Informatics Campus (SIC), Saarland University, Germany \and \url{https://cc.cs.uni-saarland.de/mblaeser/}}
{mblaeser@cs.uni-saarland.de}
{https://orcid.org/0000-0002-1750-9036}
{}

\author{Radu {Curticapean}}
{University of Regensburg, Germany \and IT University of Copenhagen, Denmark \and \url{https://www.uni-regensburg.de/informatik-data-science/algorithmen-komplexitaetstheorie/startseite/index.html} }
{radu.curticapean@ur.de}
{https://orcid.org/0000-0001-7201-9905}
{Funded by the European Union (ERC, CountHom, 101077083).
Views and opinions expressed are those of the author(s) only and do not necessarily reflect those of the European Union or the European Research Council Executive Agency.}

\author{Julian {Dörfler}}
{Saarland Informatics Campus (SIC), Saarbrücken Graduate School of Computer Science, Saarland University, Germany}
{jdoerfler@cs.uni-saarland.de}
{https://orcid.org/0000-0002-0943-8282}
{}

\author{Christian {Ikenmeyer}}
{University of Warwick \and \url{https://warwick.ac.uk/fac/sci/maths/people/staff/ikenmeyer/}}
{christian.ikenmeyer@warwick.ac.uk}
{https://orcid.org/0000-0003-4654-177X}
{Supported by EPSRC grant EP/W014882/2.}

\authorrunning{M. Bläser, R. Curticapean, J. Dörfler, C. Ikenmeyer} %

\Copyright{Markus Bläser, Radu Curticapean, Julian Dörfler, Christian Ikenmeyer} %



\category{} %

\relatedversion{} %

\nolinenumbers %

\EventEditors{John Q. Open and Joan R. Access}
\EventNoEds{2}
\EventLongTitle{42nd Conference on Very Important Topics (CVIT 2016)}
\EventShortTitle{CVIT 2016}
\EventAcronym{CVIT}
\EventYear{2016}
\EventDate{December 24--27, 2016}
\EventLocation{Little Whinging, United Kingdom}
\EventLogo{}
\SeriesVolume{42}
\ArticleNo{23}

\usepackage{todonotes}
\usepackage{amsmath}
\usepackage{xcolor}
\usepackage{amssymb}
\usepackage{mathtools}
\usepackage{url}
\usepackage{hyperref}
\usepackage{tikz}
\usetikzlibrary{patterns}
\usepackage{tikz-cd}
\usepackage{bm}
\usepackage{lineno}

\input{./macros}

\theoremstyle{claimstyle}
\newtheorem{fact}[theorem]{Fact}

\begin{document}

\maketitle

\begin{abstract}
For a fixed graph $H$, the function $\indsub{H}{\starpad}$ maps graphs $G$ to the count of induced $H$-copies in $G$; this function obviously ``counts something'' in that it has a combinatorial interpretation.
Linear combinations of such functions are called \emph{graph motif parameters} and have recently received significant attention in counting complexity after a seminal paper by Curticapean, Dell and Marx~(STOC'17).
We show that, among linear combinations of functions $\indsub{H}{\starpad}$ involving only graphs $H$ without isolated vertices, precisely those with \emph{positive integer} coefficients maintain a combinatorial interpretation.
It is important to note that graph motif parameters can be nonnegative for all inputs $G$, even when some coefficients are negative.

Formally, we show that evaluating any graph motif parameter with a negative coefficient is impossible in an oracle variant of $\sharpP$, where an implicit graph is accessed by oracle queries.
Our proof follows the classification of the relativizing closure properties of $\sharpP$ by Hertrampf, Vollmer, and Wagner (SCT'95) and the framework developed by Ikenmeyer and Pak (STOC'22), but our application of the required Ramsey theorem turns out to be more subtle, as graphs do not have the required Ramsey property.

Our techniques generalize from graphs to relational structures, including colored graphs.
Vastly generalizing this, we introduce motif parameters over \emph{categories} that count occurrences of sub-objects in the category.
We then prove a general dichotomy theorem that characterizes which such parameters have a combinatorial interpretation. Using known results in Ramsey theory for categories, we obtain a dichotomy for motif parameters of finite vector spaces as well as parameter sets.

\end{abstract}

\newpage

\section{Introduction}

\subsection{Positivity in Combinatorics and \texorpdfstring{$\sharpP$}{\#P}}
A fundamental task in combinatorics is to assign combinatorial meanings to quantities: For example, the binomial coefficient $\binom n k$ has the combinatorial interpretation of counting $k$-element subsets of $\{1,\ldots,n\}$, which may not be obvious from the algebraic formula $\frac{n!}{k!(n-k)!}$.
In fact, this is a \emph{nonnegative} combinatorial interpretation, i.e., the binomial coefficient is equal to the cardinality of a naturally defined set. This is stronger than just having a \emph{signed} combinatorial interpretation, such as for example the determinant of zero-one matrices $A=(a_{i,j})_{1\leq i,j\leq n}$, which can be expressed as the difference of the number of even versus odd permutations $\pi$ with $a_{i,\pi(i)}=1$ for all $1 \leq i \leq n$.
In the following, when we speak of combinatorial interpretations, we mean nonnegative ones.

In algebraic combinatorics, there are several nonnegative integer quantities for which finding such a combinatorial interpretation is a major open problem, e.g., problems 9, 10, 11, and 12 in \cite{Sta:00}.
This begs the question whether these quantities actually have combinatorial interpretations, but without a formal definition of this notion, their existence is hard to rule out. Indeed, Stanley \cite[Ch.~1]{Sta:11} writes that ``only through experience does one develop an idea of what is meant by a ``determination'' of a counting function''. 
In many settings however, combinatorial interpretations of a nonnegative quantity can be defined (or at least subsumed) formally in terms of containment of the quantity in the counting complexity class $\sharpP$, see \cite{IP:22, Pak:22, Pan23}. This means that the quantity can be interpreted as counting accepting paths of some nondeterministic polynomial-time Turing machine.  
While ``containment in $\sharpP$'' is a coarse over-approximation of what some mathematicians would consider ``combinatorial interpretations'', this working definition is formal, robust, and arguably natural. Most importantly, it enables us to prove that some very natural quantities indeed do \emph{not} have a combinatorial interpretation.

For example, if \ComCla{PH} does not collapse, then the squares of the symmetric group characters are not in $\sharpP$ and hence do not have a combinatorial interpretation, see \cite{IPP:23}.
Other fruitful separations from $\sharpP$ are proved in \cite{IP:22},
where the counting versions of several $\TFNP$ search problems are studied: For example, given oracle access to an exponentially large graph with degree at most two (i.e., a disjoint union of isolated vertices, paths, and cycles), and given the end of a path,
it follows that there is at least one other end of a path. But the task of finding a 
combinatorial interpretation for the nonnegative quantity ``number of other path endpoints $-1$'' fails due to the oracle separation in \cite{IP:22}.

\subsection{Do Nonnegative Graph Motif Parameters Count?}%

A graph $F$ is an \emph{induced subgraph} of $G$ if $F$ can be obtained from $G$ by deleting vertices. An \emph{induced copy} of a graph $H$ in $G$ is an induced subgraph $F$ of $G$ that is isomorphic to $H$.
For two graphs $H$ and $G$, write $\indsub HG$ for the number of induced copies of $H$ in $G$, and when $H$ is fixed, write $\indsub H\starpad$ for the function mapping $G \mapsto \indsub HG$. As defined by Curticapean, Dell, and Marx~\cite{DBLP:conf/stoc/CurticapeanDM17}, a \emph{graph motif parameter} is a rational linear combination of such functions; see Section~\ref{subsec:indsub} for formal definitions.
The step from individual counts to linear combinations turns graph motif parameters into a rich function space that captures interesting counting problems related to small patterns. For example, the number of (not necessarily induced) subgraphs isomorphic to a fixed $k$-vertex graph $H$ is a nonnegative linear combination of induced subgraph counts from all $k$-vertex supergraphs $H' \supseteq H$. 
This means that (not necessarily induced) subgraph counts from fixed graphs are graph motif parameters. 
The same holds for the number of homomorphisms from a fixed graph $H$.
Moreover, the coefficients of a graph motif parameter in its linear combination over $\indsub H\starpad$ are unique.

Graph motif parameters have been studied thoroughly before, mostly from a complexity-theoretic perspective~\cite{JerrumM15,JerrumM15b,JerrumM17,DBLP:conf/stoc/CurticapeanDM17,Roth19,Roth0W20,Roth0W21,DBLP:journals/algorithmica/DorflerRSW22,DoringMW24}:
Each fixed graph motif parameter $f$ can be evaluated in polynomial time on $n$-vertex input graphs, as it involves a linear combination of induced subgraph counts from fixed graphs.
Under complexity-theoretic assumptions, the above works give lower bounds on the exponent of this polynomial.
We focus on the separate question which graph motif parameters have combinatorial interpretations.

For example, most of the above papers consider unweighted sums (i.e., linear combinations with coefficients $0$ or $1$) of induced subgraph counts from a set $X$ of $k$-vertex graphs.
This has the evident combinatorial interpretation of counting $k$-vertex subsets $S$ of a graph $G$ that induce a graph $G[S]$ from $X$.
More generally, we can allow arbitrary integers as coefficients; as we show in Lemma~\ref{lem:graphs_integercoeffs}, all \emph{integer-valued} graph motif parameters $f$ can be obtained this way.
In this setting, it is possible for $f$ to have negative integer coefficients but still satisfy $f(G) \geq 0$ on all graphs $G$. 
\begin{example}
Consider the quantity $(|V(G)|-1)^{2} = (n-1)^{2} \geq 0$ for $n$-vertex graphs $G$. This is nonnegative and, in fact, a graph motif parameter: 
Abbreviating $\#H = \indsub HG$, we have $n = \#{K_{1}}$ and $1 = \#{K_{0}}$, and  $(n-1)^{2}=\#{K_{1}}^{2}-2\#{K_{1}}+\#{K_{0}}$. Moreover, the quadratic term $\#{K_{1}}^2$ can be linearized\footnote{
This holds because $\#{K_{1}}^2$ counts the pairs $(u,v)\in V(G)^2$, and so does the linear combination $2 \#{K_{2}} + 2 \#{I_2} + \#{K_{1}}$, whose summands correspond to the isomorphism type of $G[\{u,v\}]$, which is either $K_2$, $I_2$, or $K_1$.
As we will elaborate later, graph motif parameters enjoy a general linearization property: Any polynomial combination of induced subgraph counts can be expressed as a unique linear combination. This holds even for ``subobject'' counts under more general notions of objects; see Lemma~\ref{lem:cat:linearize} and Corollary~\ref{cor:cat:lin}. This motivates our focus on \emph{linear} combinations.}%
, i.e., expressed as the linear combination $\#{K_{1}}^{2} = 2 \#{K_{2}} + 2 \#{I_2} + \#{K_{1}}$.
Thus, we have
\[
0 \leq (\#K_1 - 1)^2 = 2\, \#K_2 + 2\,  \#I_2 - \# K_1 + \# K_0,
\]
and have constructed a linear combination of induced subgraph counts evaluating to a nonnegative value for every graph $G$, even though the term $\indsub {K_1}{\starpad}$ appears with coefficient $-1$.
\end{example}

This example in fact has a very simple combinatorial interpretation for graphs $G$ with $n \geq 1$ vertices: After fixing an arbitrary vertex $z_0$, the quantity $(n-1)^2$ simply counts the pairs $(u,v)\in V(G)^2$ with $u \neq z_0$ and $v \neq z_0$.
Let us consider a example where finding such an interpretation seems trickier:
\begin{example}
\label{ex:nonneg}
Consider the graph motif parameter
\begin{equation*}
f(G) \ = \ \#\Gline - \#\Gtria + \#\GfourI + 2\,\#\GfourII + 4\,\#\GfourIII \ \geq \  0,
\end{equation*}
whose non-negativity is shown in \S\ref{sec:nonnegexample}.
Contrary to our previous example, the graphs in this linear combination feature no isolated vertices.
Does $f$ have a combinatorial interpretation?
\end{example}
Our main result in Theorem~\ref{thm:graphs} gives a formal argument against $f$ having a combinatorial interpretation:
In a nutshell, a nonnegative graph motif parameter involving only patterns without isolated vertices has a combinatorial interpretation if and only if all its coefficients are nonnegative. (The ``if'' part of this result is trivial.)
In the next section, we will give the relevant formal definitions to state our results in Section~\ref{sec:our-results} and compare them to related work in Section~\ref{sec:related-work}.

\section{Preliminaries}
\label{sec:prelim}
\subsection{Type-2 Counting Complexity} \label{subsec:typetwo}

We study computational problems on inputs given by concise descriptions. Such problems have been studied by Johnson, Papadimitrou, and Yanakakis \cite{DBLP:journals/jcss/JohnsonPY88, DBLP:journals/jcss/Papadimitriou94},
who introduced subclasses of the complexity class $\TFNP$, such as $\PLS$ or $\PPAD$, to capture search problems in exponentially large objects.
The problems in these subclasses of $\TFNP$ all have a similar structure: Each input $x$ implicitly encodes an exponentially large structure, e.g., a graph~$G$ via a Boolean circuit $x$ that computes the edge relation of the graph $G$. %
The witness sought in the search problem is however still polynomially large in $|x|$, which is tiny compared to the size of the encoded input instance.
The guiding examples in our paper are induced subgraphs $H$ of a fixed size in an exponentially large graph $G$.
In this problem, the $O(1)$ labels of the vertices in $G$ that induce $H$ constitute a witness.

A natural abstraction of this setting is obtained by providing the input as an oracle, rather than via a succinct encoding, see e.g.\ \cite{DBLP:conf/focs/BeameIP95}. In this setting, which is denoted as \emph{type-2 complexity}, we can only query the given structure, for instance by making queries to the edge relation of the graph, but we do not have access to a circuit computing the edge relation. 
Type-2 problems also arise naturally in computability and complexity studies of real functions,
see e.g.\ \cite{DBLP:series/txtcs/Weihrauch00, DBLP:books/sp/91/K1991}. Since a real number is an infinite object, it is given as an oracle which one can query to get better and better rational approximations.

To define the type-2 framework formally,
fix some alphabet $\Sigma$.
We assume that $0,1 \in \Sigma$, which are interpreted as false and true, respectively.
An \emph{oracle} is a total function of type $\Sigma^\ast \to \{0,1\}$. 
The computational model in the type-2 framework is an oracle Turing machine, which gets an input string $x$ together with access to an oracle $O$. %

\begin{example}
For a fixed graph $H$, we consider the problem of finding an induced $H$-copy in a large graph $G$. The graph $G$ is given as $(x, O)$ where $O: \Sigma^\ast \to \{0,1\}$ is an oracle encoding the edge relation and $x$ is the number of nodes of $G$ written in binary. The nodes of $G$ are $0,\dots,x-1$, written as bit strings of fixed length~$n$.
Given two nodes $u$ and $v$ of $G$, there is an edge between $u$ and $v$ iff $O(uv) = 1$. (Since we consider strings of fixed length $n$, no separation symbol is needed.) 
In this representation, most of the information about $G$ is stored in the oracle; the input $x$ only encodes the size of the graph, and the running time of oracle Turing machines is measured in terms of $x$.
\end{example}

A consistency issue can arise in this definition: 
Since the graph $G$ is exponentially large, the oracle Turing machine might not be able to check whether the encoding is consistent. For instance, one might wish to encode an undirected graph $G$ but have $O(uv) \neq O(vu)$. In the case of graphs, such inconsistencies can be ``repaired'' by avoiding observing them, e.g., by only querying strings $uv$ with $u < v$.
In the context of $\TFNP$ and its subclasses, such issues are often elegantly repaired by interpreting any oracle as a graph with the desired properties.
In general however, such a fix may not be possible for more complicated structures, in particular in the abstract setting of category theory that we will also study. 
Therefore, we will also resort to promise problems, in which a Turing machine is only required to work properly on inputs $(x,O)$ that are a proper encoding of an input object; see Definition~\ref{def:promise-sharpp} for more details. Note that this also encompasses the previous solutions. If we can repair the oracle or choose an interpretation such that every oracle encodes a valid instance, then we simply choose the set of ``don't care'' instances to be the empty set.

A $k$-ary \emph{type-2 relation} $R$ gets a tuple of $k$ strings $x$ as an input, has access to an oracle $O$, and outputs a value $R(x,O) \in \{0,1\}$. Typically, we will consider unary type-2 relations, but relations with higher arity naturally occur when one considers witnesses, see below, when we define the existential and counting operators. %
A type-2 relation $R$ is polynomial-time computable if there is an oracle Turing machine $M$ that given input $x$ and oracle access to $O$ computes $R(x,O)$ in polynomial-time in the length of the input $x$. Oracle queries have unit costs, but $M$ has to write the query strings on the oracle tape. Note that all time bounds we consider depend only on the length $|x|$ of $x$, but not on $x$ itself and, more importantly, not on the oracle.%

If there is a polynomial-time computable $(k+1)$-ary type-2 relation $S$, we define the 
$k$-ary type-2 relation $\exists S$ by $\exists S(x,O) = 1$ iff there is a polynomially long $y$ (in the length $|x|$) such that $S(x,y,O) = 1$. 

\newcommand{\IsIndSub}{\mathrm{IsIndSub}}

\begin{example}
In our example of finding an induced copy of $H$ in the graph $G$, we start with a binary
type-2 relation $\IsIndSub_H(x,y,O)$. The relation interprets $y$ as an encoding of a set $Y$ of $k$ distinct nodes of $G$ with $k$ being the number of nodes of $H$. 
The relation then queries the oracle $O$ to obtain the induced subgraph $G[Y]$.
It then checks by brute force whether $G[Y]$ and $H$ are isomorphic. If yes, it returns the value $1$ and otherwise $0$. This is clearly polynomial-time computable in $|x|$. The relation $\exists \IsIndSub_H(x,O)$ now returns $1$ if $G$ contains an induced copy of $H$ \emph{somewhere}, and $0$ if $G$ does not contain such a copy. See also Lemma~\ref{lem:graphs_lower} for a more formal proof.
\end{example}

As in classical type-1 complexity, we can also introduce a counting operator $\#$ to type-2 complexity, similar to the existential operator $\exists$:
A polynomial-time computable $(k+1)$-ary type-2 relation $S$ defines a
$k$-ary type-2 counting problem by $\#S(x,O) = \# \{y \mid S(x,y,O) = 1 \}$,
where $y$ is polynomially long.

\begin{example}
Continuing the example above, $\#\IsIndSub_H$ now counts the number of induced subgraphs in $G$ that are isomorphic to $H$, that is, $\#\IsIndSub_H(x,O) = \indsub H G$, where $(x,O)$ is the oracle encoding of $G$.
\end{example}

The goal of this paper is to characterize those integer valued nonnegative graph motif parameters that have a combinatorial interpretation.
In other words, we ask which nonnegative graph motif parameters count something. The informal notion of ``combinatorial interpretation'' is formalized below as follows:

\begin{definition}
   $\sharpPdottedcircle = \{ \#R \mid \text{$R$ is a polynomial-time computable binary type-2 relation} \}$.
\end{definition}

In the notation $\sharpPdottedcircle$, the dotted circle represents a placeholder for the input oracle.

\begin{definition}\label{def:combinatoriallyinterpretable}
For a graph motif parameter $\varphi$, we denote by 
$\eval \varphi$ the corresponding type-2 function, which gets the graph $G$ given in the form $(x,O)$.
An integer valued nonnegative graph motif parameter $\varphi$
is called \emph{combinatorially interpretable} if $\eval \varphi \in \sharpPdottedcircle$ (depending on context also $\eval \varphi \in \promise\sharpPdottedcircle$, see later for an explanation).
\end{definition}

It is open for discussion whether every problem in $\sharpPdottedcircle$ admits a ``nonnegative combinatorial interpretation'' in the informal sense used by combinatorists. 
If a problem is however \emph{not} in $\sharpPdottedcircle$, then it cannot have a nonnegative combinatorial interpretation. 
Our dichotomy result for integer-valued nonnegative graph motif parameters is particularly clean-cut in this regard: Such a parameter is in $\sharpPdottedcircle$ iff all its coefficients are nonnegative integers, so it indeed counts induced subgraphs.

To study more complicated structures, we consider promise problems, since the oracle TM might not be able to check whether an oracle is a legal encoding. Hence, we will have a set of ``don't care'' inputs $(x,O)$, which will correspond to illegal encodings of input objects in our setting. Our oracle TM only needs to compute the correct result on inputs which are not ``don't care'' inputs. One such example are vector spaces, where we interpret the oracle as the characteristic function of the given vector space and it is not immediate how to check whether the oracle indeed describes a vector space. In this case our oracle Turing machine only has to work properly if the given oracle encodes indeed a vector space.

\begin{definition}
\label{def:promise-sharpp}
A counting function $F$ together with a set $D$ of ``don't care'' inputs is in $\promise\sharpPdottedcircle$ if there is a  polynomial-time computable binary type-2 relation $R$ such that $\#R$ and $F$ coincide on all inputs outside $D$.
\end{definition}

Every function in $\sharpPdottedcircle$ is also in $\promise\sharpPdottedcircle$ with the empty set of ``don't care'' inputs.

\subsection{What is a Combinatorial Interpretation---and What is Not?}
\label{sec:combint}

While we do not have an all-encompasing answer to the first part of this question (and we do not need one in this paper), we explain why type-2 complexity gives satisfactory answers to the second part.   Intuitively, a combinatorial interpretation should mean ``counting something''. But that is not enough. Even the graph motif parameter $f$ from Example~\ref{ex:nonneg} counts something: For a given graph $G$, we consider the set of witnesses for each pattern in $f$. So $W(\Gline)$ are all (induced) occurrences of an edge in $G$, $W(\Gtria)$ are all (induced) occurrences of a triangle in $G$ and so on.
The sets $W(\GfourII)$ and $W(\GfourIII)$ will be multi-sets containing two and four copies of the witnesses, respectively, due to the coefficients. Let $W_+$ be the union of all witness sets that correspond to a term with positive coefficient and $W_- = W(\Gtria)$ be all witnesses that correspond to a term with negative coefficient. Since $f$ is nonnegative, $|W_+| \ge |W_-|$.
Thus there is an injective map $i: W_- \to W_+$. So we could say that $f$ counts all elements in $W_+$ that are not in $\im(i)$. However, this ``cheating away'' of the difference $|W_+| - |W_-|$ is unsatisfactory, because it is (computationally) hard to decide whether a given element in $W_+$ should be counted or not.

Counting complexity gives a way to define combinatorial interpretations: We are given a set of witnesses $W \subseteq \{0,1\}^*$, encoded as binary strings, and for each $x \in \{0,1\}^*$ we can decide in polynomial time whether it is in $W$ or not. This rules out the construction above, since constructing the embedding $i$ is infeasible. Computing the size $|W|$ is a problem in the complexity class $\sharpP$: A nondeterministic Turing machine guesses a witness $x$ and then deterministically verifies whether it is in $W$. The number of accepting paths is precisely $|W|$.  
While it might be debatable whether being in $\sharpP$ is the right definition for a combinatorial interpretation, it appears to be at least a necessary condition. In other words, if a quantity is not in $\sharpP$, then it does not admit a combinatorial interpretation.

In the case of graph motif parameters, the witnesses are given as ordered lists of nodes having constant length, say $k$, since only the target graph $G$ varies. If $G$ has $N$ nodes, each node can be encoded by a bitstring of length $\log N$ and the total witness size is $k \log N$. When the running time should be polynomial, then the Turing machine $M$ for checking the witness needs random access to the graph, because it cannot inspect the whole graph. This is modeled using oracles and type-2 complexity.

\subsection{Induced Subgraph Numbers} \label{subsec:indsub}
Graphs in this paper are finite, undirected, and simple, i.e., they do not contain multi-edges or self-loops. Given a graph $G$, we write $V(G)$ and $E(G)$ for its vertex and edge set, respectively.
We write $\graphs$ for the class of all graphs.
Towards applying a Ramsey theorem, we will also consider \emph{ordered} graphs: These are graphs $G=(V,E,\leq_G)$ where $(V,E)$ is a graph and $\leq_G$ is a total order on $V$.
We write $\ordgraphs$ for the class of these graphs.
When ordered graphs are represented by an oracle, we always assume the natural lexicographic order on the vertices.

Given graphs $H$ and $G$, we say that $H$ is a \emph{subgraph} $H\subseteq G$ of $G$ if $V(H) \subseteq V(G)$ and $E(H) \subseteq E(G)$. In the case of ordered graphs $H=(V,E,\leq_H)$ and $G=(V',E',\leq_G)$, we additionally require that $\leq_H$ is the restriction of $\leq_G$ to $V$. 
Given a vertex set $X \subseteq V(G)$, write $G[X] = (X,\{vw \in E(G) \mid v,w \in X\})$ for the \emph{subgraph} of $G$ \emph{induced} by $X$. 
We write $H \sqsubseteq G$ to denote that $H$ is an induced subgraph of $G$, and we write $\poset{G}$ for the poset of induced subgraphs of a graph $G$, ordered by $\sqsubseteq$.
We call $H$ \emph{pure} if it contains no isolated vertices, i.e., no connected components isomorphic to $K_1$, and we write $\graphsnoniso$ for the class of pure graphs and $\ordgraphsnoniso$ for the pure ordered graphs.

A \emph{homomorphism} from a graph $H$ into a graph $G$ is a function $f:V(H)\to V(G)$ such that $uv\in E(H)$ implies $f(u)f(v) \in E(G)$; we write $\hom{H}{G}$ for their number.
A homomorphism is an \emph{embedding} if it is injective; we write $\hominj{H}{G}$ for their number.
A \emph{strong embedding} is an embedding with the additional condition that $uv \notin E(H)$ also implies $f(u)f(v) \notin E(G)$; we write $\strhominj{H}{G}$ for their number.
An \emph{isomorphism} is a strong embedding that is bijective; two graphs $H$ and $G$ are isomorphic if there is an isomorphism from $H$ to $G$.
For isomorphic graphs $H, G$ we write $H \cong G$.
All of these definitions apply to ordered graphs as well; in this setting, homomorphisms $f$ must additionally preserve the orderings of $H$ and $G$, i.e., if $u \leq_H v$, then $f(u) \leq_G f(v)$.

An automorphism of $H$ is a strong embedding from $H$ into $H$;
we write $\aut{H}$ for the numbers of automorphisms of $H$.
The numbers of embeddings and strong embeddings from a graph $H$ into some graph $G$ are multiples of $\aut{H}$. We obtain the number of subgraphs isomorphic to $H$ as $\sub{H}{G} = \hominj{H}{G} / \aut{H}$ and the number of induced subgraphs isomorphic to $H$ as $\indsub{H}{G} = \strhominj{H}{G} / \aut{H}$.
Note that for ordered graphs $H$, we have $\aut H = 1$.
For a graph $H$, we write $\indsub{H}{\starpad}$ for the function that maps $G \mapsto \indsub{H}{G}$. 
We use that these functions are linearly independent for non-isomorphic graphs $H$, even when restricted to particular domains:
\begin{fact}
\label{fact: lin-indep-ind}
Let $\graphs' \subseteq \graphs$ be a finite set of pairwise non-isomorphic graphs.
For $H\in \graphs'$, write $f_H : \graphs' \to \mathbb Q$ for the restriction of $\indsub{H}{\starpad}$ to $\graphs'$.
Then the functions $\{ f_H \mid H \in \graphs' \}$ are linearly independent.
The same applies for ordered graphs. 
\end{fact}

Having established that the functions $\indsub{H}{\starpad}$ are linearly independent, we consider linear combinations of these basis functions.
The following term was coined in~\cite{DBLP:conf/stoc/CurticapeanDM17}:
\begin{definition}
\label{def:graph-motif-parameter}
A \emph{graph motif parameter} is a function $\varphi : \mathcal G \to \qqq$ that admits pairwise non-isomorphic pattern graphs  $H_1,\ldots,H_s$ and coefficients $\alpha_1 ,\ldots,\alpha_s \in \qqq$ such that, for all graphs $G$,
\begin{equation}
\label{eq: graphmotif-lincomb}
\varphi(G)= \sum_{i=1}^s \alpha_i \cdot \indsub{H_i}{G}.
\end{equation}
We say that $\varphi$ is \emph{pure} if every graph $H_i$ with non-zero coefficient is pure.
An \emph{ordered graph motif parameter} similarly is a linear combination of induced subgraph counts from pairwise non-isomorphic \emph{ordered} graphs.
The \emph{support} $\supp(\varphi)$ is the set of all $H_i$, s.t.\ $\alpha_i \neq 0$.

\end{definition}
Note that $\varphi$ being pure is well-defined, as the linear combination \eqref{eq: graphmotif-lincomb} is unique by Fact~\ref{fact: lin-indep-ind}:
The (isomorphism types of) graphs and coefficients appearing in the linear combination are uniquely determined; we therefore speak of \emph{the coefficients} and \emph{the patterns} of $\varphi$.
Since we investigate whether or not a graph motif parameter $\varphi$ has a combinatorial description, i.e., whether $\varphi$ counts a set of objects, we will restrict ourselves to graph motif parameters with image $\IN$, as all others can clearly not count anything. By the following lemma, this implies that the coefficients in the linear combination \eqref{eq: graphmotif-lincomb} can be assumed to be integers.

\begin{lemma}
    \label{lem:graphs_integercoeffs}
    Let $\varphi$ be an (ordered) graph motif parameter.
    Then $\varphi(G) \in \IZ$ for every (ordered) graph $G$ if and only if all coefficients of $\varphi$ are integers.
\end{lemma}
\begin{proof}
    The ``if'' is obvious, since $\indsub{H}{G}\in \IZ$ for any $H,G$.
    For the ``only if'', let $H$ be a pattern in $\varphi$ 
    with a coefficient $\alpha_H \in \IQ \setminus \IZ$, such that $|V(H)|$ is minimal among all such patterns.
    Then $\indsub{H}{H} = 1$ and $\varphi(H) = \alpha_H + \beta \notin \IZ$ for some $\beta \in \IZ$, since all other basis functions with non-integer rational coefficients evaluate to $0$ on $H$, as their pattern has at least  $|V(H)|$ vertices and is not isomorphic to $H$.
\end{proof}

\begin{lemma}
\label{lem:graphs_lower}
Let
$
\varphi(G)= \sum_{i=1}^s \alpha_i \cdot \indsub{H_i}{G}
$
with all $\alpha_i\in\IN$.
Then $\eval{\varphi} \in \sharpPdottedcircle$.
This holds for unordered and ordered graph motif parameters.
\end{lemma}

\section{Our Results}
\label{sec:our-results}
\subsection{Graphs}
We characterize the pure graph motif parameters $\varphi$ with a combinatorial interpretation as those whose coefficients are nonnegative integers. Such graph motif parameters $\varphi$ clearly have a combinatorial interpretation, as they count all induced occurrences of pattern graphs, possibly with multiplicities (see Lemma~\ref{lem:graphs_lower}.) On the other hand, if $\varphi$ has a negative coefficient, then we show that $\varphi$ is not in $\promise\sharpPdottedcircle$ and thus does not admit a combinatorial interpretation.
Formally we prove (recall Def.~\ref{def:combinatoriallyinterpretable} about combinatorial interpretability):

\begin{theorem}
\label{thm:graphs}
Let $\varphi$ %
be a pure graph motif parameter, i.e.,
$\varphi(G)= \sum_{i=1}^s \alpha_i \cdot \indsub{H_i}{G}$
for pairwise non-isomorphic pure $H_1,\ldots,H_s$ and coefficients $\alpha_1 ,\ldots,\alpha_s \in \qqq$.\footnote{By Lemma~\ref{lem:graphs_integercoeffs} the interesting cases are those where all $\alpha_i\in\IZ$.}
Then $\varphi$ is combinatorially interpretable
iff all $\alpha_1 ,\ldots,\alpha_s \in \qqq$ are nonnegative integers.
\end{theorem}
We show Theorem~\ref{thm:graphs} in \S\ref{sec:graphs} with a delicate and technically demanding proof. The intuition is however rather easy to explain, and we illustrate it on the Example~\ref{ex:nonneg}.
An NTM $M$ is supposed to output $f(G)$, but it can only query a small part of $G$.
Assume that $G$ is either the union of a triangle and exponentially many singletons, or the union of an edge with exponentially many singletons, or $G$ is the graph that consists of only exponentially many singletons.
If $G$ is an edge and singletons, then one computation path of $M$ must accept and must have queried the oracle at the edge, because otherwise this computation path would lead to an accepting path even if $G$ has no edges, which would be the wrong answer.
Now, in the case that $G$ is a triangle and singletons, it turns out (which is delicate) that such an accepting path exists for each of the triangle edges, so we get 3 accepting computation paths, but $f(G)=2$, so the NTM has too many accepting paths.

\subsection{More General Structures}

The methodology used to prove Theorem~\ref{thm:graphs} can be generalized to more general structures, provided that they satisfy the following conditions:
\begin{itemize}
    \item Ramsey property: We require Ramsey-type theorems for the structures, which assert---roughly speaking---that for given objects $A,B$, there exists an object $C$ which contains enough $B$-copies such that, by partitioning the set of all $A$-copies in $C$ into a small number of parts, there is some part that contains a full $B$-copy.   
    In the case of pure graph motif parameters, for instance, this is Lemma~\ref{lem:nesetril}: For any ordered graphs $F, G$, there is an ordered graph $H$ such that for any partitioning of the induced subgraphs of $H$ that are isomorphic to $F$, we can find an induced copy of $G$ in $H$ such that all induced copies of $F$ in this copy of $G$ are in the same part.
    \item A neutral padding operation: We require an operation to enlarge a given object without increasing the number of induced subobjects.
    In the case of pure graph motif parameters, this can be achieved by adding isolated vertices, however we introduce different possible ways of achieving this for different objects, leading to different definitions of ``pure''.
\end{itemize}

\subsubsection{Relational Structures}

For instance, we can prove results similar to Theorem~\ref{thm:graphs} for various variants of relational structures, which generalize graphs. In graphs, we consider exactly one binary relation on a set $V$, the edge relation. In a relational structure, we can have multiple relations of various arities. A relational structure $A$ is an induced substructure of another relational structure $B$,
if we can obtain $A$ by restricting the relations of $S$ to the vertices of $A$.
Similarly to graphs, we can now define \emph{motif parameters} as linear combinations of induced substructure numbers.
Formal definitions are given in Section~\ref{sec:relational_structures}.
We get the following classification for motif parameters of relational structures.

\begin{theorem}[informal, see Theorem~\ref{thm:ordrelstructures} for a precise statement]
The evaluation function of a pure motif parameter of relational structures is combinatorially interpretable (i.e., it is in $\promise \sharpPdottedcircle$) iff all its coefficients are nonnegative integers.
\end{theorem}

We also consider further types of relations, namely multiset relations (each relation is a multiset and can have repetitions), and list relations (which is the equivalent to directed graphs in the relational setting) with and without repetitions. We call these relational structures \emph{mixed}. Precise definitions are again given in Section~\ref{sec:relational_structures}.
For such mixed relational structures, we get similar dichotomies, which we first prove in the ordered setting (Theorem~\ref{thm:ordmixedrelstructures}), since the corresponding Ramsey property only holds for ordered mixed relational structures. 
Then we transfer this theorem to the unordered setting in Theorem~\ref{thm:mixedrelstructures}---a similar detour into the ordered setting is in fact already needed for the case of graphs.
    
An application of mixed relational structures are colored graphs, where nodes are colored and isomorphisms have to respect the colors. Theorem~\ref{thm:colored_graphs} shows a dichotomy result for colored graphs.

\subsubsection{Category Theory}

Ramsey-type theorems have even been established within category theory \cite{graham1972ramsey, DBLP:journals/jct/Nesetril77}. 
Since Ramsey arguments are crucial in our method, this raises the interesting question whether we can also transfer our results into this realm. Indeed, we achieve this in Section~\ref{sec:category_general}. 
Our objects of consideration are now objects of some category $C$, and the subobjects we consider are now specified using a class of morphisms $\cM$ from $C$.
For two objects $a$ and $b$ of $C$, we call the morphisms $a \to b$ in $\cM$, identified if they only differ by an isomorphism of $a$, the \emph{$\cM$-subobjects} of $a$ under $b$. In the case of undirected graphs, we choose $\cM$ to be the class of all strong graph embeddings.
An important concept in our construction will be so-called factorization systems, used also by Lagodzinski~\cite{Lagodzinski24} and Isbell~\cite{ISBELL199187} in the context of counting problems.
We have a second class of morphisms $\cE$ and we require that every morphism $f: a \to b$ factors (essentially uniquely) as $f = me$ with $m \in \cM$ and $e \in \cE$.
In our running example of unordered graphs, we can take $\cE$ to be the class of surjective graph homomorphisms and $\cM$ the above class of strong graph embeddings.
Then every graph homomorphism $f: H \to G$ factors through the subgraph $G[f(V(H))]$ of $G$ that is induced by the image of $V(H)$. 

Recall that for graph motif parameters, we had a notion of pure graphs. We will have the same here: We will have two categories, a category $C$ of objects that we consider and a category $P$ of pure objects that we want to count. In Lemma~\ref{lem:cat:prop}, we collect a number of natural and rather mild requirements on $C$ and $P$ (maybe except for the property of being Ramsey) in order to prove our main result that $P$-pure motif parameters can only have a combinatorial interpretation if the coefficients are all nonnegative.

\begin{theorem}[informal, see Theorem~\ref{thm:cat_main_theorem} for the precise statement]
    Let $\varphi$ be a $P$-pure motif parameter.  
    Then under the conditions specified in Lemma~\ref{lem:cat:prop}, we have $\eval{\varphi} \notin \promise\sharpPdottedcircle$ unless all its coefficients are nonnegative integers.
\end{theorem}

The theorem does not automatically give a dichotomy, since we do not have a corresponding upper bound when the coefficients are all nonnegative integers. However, our concrete applications in fact do admit corresponding upper bounds:
\begin{itemize}
    \item 
    First we consider finite dimensional vector spaces over finite fields, which were shown to have the Ramsey property in \cite{graham1972ramsey} using methods from category theory. 
    In particular, we show that the evaluation function of a motif parameter over a finite vector space is in $\promise\sharpPdottedcircle$ iff all its coefficients are nonnegative integers (Theorem~\ref{thm:cat:main:vectorspace}).
    \item Then we consider so-called parameter sets, which are subsets of $A^n$ for some fixed alphabet $A$ where we may set a coordinate to constant or identify any number of coordinates. They are also known to satisfy a Ramsey property, see \cite[Theorem 10.4]{Nesetril1995}.
    As a consequence of our general theorem, 
    a motif parameter over parameter sets is in $\promise\sharpPdottedcircle$ iff all its coefficients are nonnegative integers (Theorem~\ref{thm:cat:main:parsets}).
\end{itemize}

\subsection{Linearization}

As shown examplarily on the first pages, graph motif parameters enjoy a useful linearization property: Any polynomial in induced subgraph counts is a graph motif parameter, that is, a \emph{linear} combination of induced subgraph numbers.
Among other applications, this universality of linear combinations can be used to show that induced subgraph counts are algebraically independent, as any annihilating polynomial would in fact translate to a linear combination that always evaluates to zero, thus contradicting linear independence.
Linearization also motivates our focus on linear combinations of subobject counts.

In Lemma~\ref{lem:cat:linearize} and Corollary~\ref{cor:cat:lin}, we establish conditions for subobject numbers of a category $C$ to enjoy similar linearization properties.
We obtain that the coefficients in the linearized version are positive and have a combinatorial interpretation if the polynomial we started with has positive coefficients in the binomial basis; see the next section. 

\section{Related Work}
\label{sec:related-work}
Searching for combinatorial interpretations for nonnegative integer quantities has a long tradition in combinatorics dating back to Cayley, Littlewood, Richardson, Sch\"utzenberger, Macdonald, Schensted, Knuth, Stanley, and others, e.g.,~\cite{K1882, Sch77, Sta:00}, and has recently been studied from the perspective of $\sharpP$ containment and functional $\sharpP$ closure properties \cite{Pak19, IP:22, IPP:23, CP24}.
The univariate relativizing functional closure properties of $\sharpP$ have been characterized in \cite[Thm~3.1.1(b)]{CGHHSWW:89}:
Consider the binomial coefficient $\binom{x}{i}=\frac{1}{i!} x (x-1) (x-2) \cdots (x-i+1)$ for $i \in \IN$ as a polynomial in $x$, and let $\varphi:\IN\to\IN$ be a univariate polynomial with expansion $\varphi(x) = \sum_{i=1}^k \alpha_i \binom{x}{i}$ for $\alpha_i \in \IQ$.
This expansion is called the binomial basis expansion of $\varphi$.
It is known that $\varphi(\sharpPdottedcircle) \subseteq \sharpPdottedcircle$ iff all $\alpha_i\in\IN$ (see \cite{IP:22} for a proof).
The connection between combinatorial interpretations and functional closure properties of $\sharpP$ is very recent \cite{IP:22}.

In the language of graphs, the construction in \cite{CGHHSWW:89} can be phrased as follows.
Let $\graphs'$ denote the graphs on $\{1,2,\ldots, 3k\}$ whose vertices $\{3i,3i+1,3i+2\}$ either form a triangle or a path $\Gpath$ on 3 vertices.
Consider graph motif parameters for which all patterns are disjoint unions of triangles, i.e.,
$\varphi(G) = \sum_i \alpha_i \indsub{H_i}{G}$, where $H_i$ is a union of $i$ disjoint triangles.
For $G\in\graphs'$ with $x$ many triangles, we have $\varphi(G) = \sum_i \alpha_i \binom{x}{i}$. Their construction can be interpreted as oracle graphs $G\in\graphs'$ showing that $\varphi$ is not combinatorial.

The non-oracle implications in the univariate case were studied in \cite[p.~307]{OH:93}, which shows that $\ComCla{GapP}\cap\{f:\{0,1\}^*\to\IN\mid \forall w: f(w)\geq 0\} = \sharpP$ implies $\ComCla{C}_=\ComCla{P} = \ComCla{coNP}$ and $\ComCla{SPP} = \ComCla{UP}$.\footnote{See \cite{DBLP:series/txtcs/HemaspaandraO02} for an overview over these classes. You will however not need to know about them to understand this paper.} The multivariate case was resolved in \cite[Thm~3.14]{HVW:95}, cf.~\cite[Thm~23]{Her:95}, based on a Ramsey theorem \cite[Thm~12]{Her:95}.
As in the univariate case, one gets small variants of our main theorem, whereas we consider a larger set of non-isomorphic graphs of the same cardinality, each without isolated vertices.
The functional closure properties of other counting classes are mainly unknown, with the exception of $\GapP$ \cite[Thm~6]{Bei:97} and finite automata \cite{DI:24}.

While search versions of type-2 problems have been studied already in \cite{BCEIP:95} and later works on $\TFNP$, the counting versions have only recently been studied to understand whether their nonnegative transformations still have combinatorial interpretations \cite{IP:22}.
We improve on the work \cite{IP:22} in various ways. First, we show much more general results due to our general theorem in the categorial framework. Second, our results are also valid for unordered objects, like unordered graphs, which do not have the required Ramsey property per se. We overcome this technical difficulty by carefully reducing the unordered case to the ordered one.
All graph problems in \cite{IP:22} are on ordered graphs only.

\section{Results for Graphs}
\label{sec:graphs}
\subsection{Reduction to Ordered Graphs}
To prove Theorem~\ref{thm:graphs}, we first prove the following stronger result about ordered graphs:

\begin{theorem}
\label{thm:ordgraphs}
Let $\varphi: \ordgraphs \to \IN$ be a pure ordered graph motif parameter.
Then $\eval{\varphi} \in \promise\sharpPdottedcircle$ iff all coefficients of $\varphi$ 
are nonnegative integers.
\end{theorem}

Theorem~\ref{thm:graphs} then follows by the simple observation that 
\begin{align}
    \label{eq:obs_symmetrize}
    \indsub{(V_i, E_i)}{(V, E)} &= \sum_{\leq_{V_i}}\indsub{(V_i, E_i, \leq_{V_i})}{(V, E, \leq_V)}
\end{align}
where $\leq_V$ is any arbitrary linear ordering of $V$ and $\leq_{V_i}$ sums over all linear orders of $V_i$ that result in non-isomorphic ordered graphs.
Note that all the pattern graphs occurring on the right-hand side are non-isomorphic pure ordered graphs.
Even more, for different patterns on the left hand side, all possible occuring patterns on the right hand side are non-isomorphic.
Now let $\varphi: \graphs \to \IN$ be a pure graph motif parameter with a negative coefficient.
Then the pure ordered graph motif parameter $\varphi': \ordgraphs \to \IN$ obtained by applying \eqref{eq:obs_symmetrize} to each of the basis functions has a negative coefficient and thus $\eval{\varphi'} \notin \promise\sharpPdottedcircle$ by Theorem~\ref{thm:ordgraphs}.
Assume for the sake of contradiction that $\eval{\varphi} \in \promise\sharpPdottedcircle$, then we can compute $\eval{\varphi'}$ by simulating the ordered oracle in polynomial time and running the corresponding NTM computing $\eval{\varphi}$.
This places $\eval{\varphi'}$ in $\promise\sharpPdottedcircle$, a contradiction.

While this only proves that pure graph motif parameters with some negative coefficients are not combinatorially interpretable, Lemmas~\ref{lem:graphs_integercoeffs} and~\ref{lem:graphs_lower} complete the proof by proving the reverse direction.

\subsection{Ordered Graphs}
The proof of the hardness in Theorem~\ref{thm:ordgraphs} now proceeds in two parts:

First we show how to instantiate any ordered graph and all its induced subgraphs as oracle graphs.
Here it is important that the instantiations all have the same size; otherwise, the Turing machine could infer (through the number of vertices) possibly useful information about an instantiation without explicitly querying the oracle.
For this purpose, for some subset $D \subseteq \ordgraphs$, we say that an ordered graph motif parameter $\Psi: D \to \IN$ is \defn{$D$-good} if all its coefficients are nonnegative integers and all its patterns 
are from $D$.
If any such coefficient is negative, it is instead called \defn{$D$-bad}.
If the domain $D$ is clear from the context, we will also simply call them \defn{good} and \defn{bad} respectively.

The second part of the proof finds a target graph $G$ such that, if there is any counterexample where any (restricted) good function would disagree with our bad $\varphi$, it must already occur on some induced subgraph of $G$.
For a fixed bad pure ordered graph motif parameter $\varphi: \ordgraphs \to \IN$ such a counterexample will always exist.
We then instantiate precisely $\poset{G}$ as oracle graphs.

Simply padding all graphs of $\poset{G}$ to the same size however is not enough.
The main goal here is that acceptance of a computation path is monotone relative to which part of the graph is observed by a computation, i.e., if we embed $H \sqsubseteq G$, then all previously accepting paths that accept on oracle input $H$ should still be accepting when the oracle input is $G$.

Our proof follows the structure of \cite{IP:22}.
We start by formalizing computation paths:
\begin{definition}\label{def:computpath}
A \defn{computation path} $\tau$ of a nondeterministic Turing machine
on some input is defined as the sequence of its nondeterministic choice bits and the answers to its oracle queries (both types of bits appear in the same list, ordered chronologically).
Formally, it is an element of $\{0,1\}^*$.
\end{definition}

The same Turing machine can yield the same computation path on different inputs, for example, when not the whole input is read, or when having access to different oracles, because the oracles can differ in positions that are not queried.

Given a nondeterministic Turing machine $M$ and an oracle graph $G$,
we are interested in the number of accepting paths of $M$ when given oracle access to $G$.
We define the number of accepting paths of $M$ on input $j \in \IN$ (given in binary) with oracle access to $G$ as $\#\acc_{M^G}(j)$. (This notation highlights that $G$ is given as an oracle while the number of vertices $j$ is a classical input.)
Further we will also call $\#\acc_{M^G}(j)$ a computation.
In the following, write $\ordgraphssize{j}$ to denote those graphs in $\ordgraphs$ that contain exactly $j$ vertices, $0, \ldots, j-1$ with the natural order on them, for $j\in \nn$.

\begin{definition}[Set-instantiator against $(M,G,\varphi)$]
\label{def:graphs_setinstantiator}
Let  $M$  be a nondeterministic Turing machine,
let $G \in \ordgraphs$ and
let $\varphi$ be an ordered pure graph motif parameter.
Let  $\top$  be a symbolic top element above the induced subgraph poset $\poset{G}$, i.e., $\top \not\sqsubseteq H$ for all induced subgraphs $H \sqsubseteq G$.
A  \defn{set-instantiator}  $\SI$  is a triple of

$\circ$ \ some $j \in \IN$, 

$\circ$ \ an instantiation function  $\inst_\SI : \poset{G} \to \ordgraphssize{j}$,   and

$\circ$ \  a perception function  $\perc_\SI : \{0,1\}^* \to \poset{G} \cup \{\top\}$,

\smallskip
\noindent
such that both of the following properties hold for all induced subgraphs $H \sqsubseteq G$:

\smallskip
$\bu$ \ $\tau\in\{0,1\}^*$ is an accepting path in computation $\#\acc_{M^{\inst_\SI(H)}}(j)$
iff $\perc_\SI(\tau) \sqsubseteq H$,

$\bu$ \ $\varphi(\inst_\SI(H))) = \varphi(H)$.
\end{definition}

In the above definition, we think of the perception of an accepting computation path of $M$ to be the induced subgraph of $G$ that is observed by $M$ through oracle queries, while computation paths that do not accept are given perception $\top$.

We can now prove by a randomized construction that set-instantiators always exist for polynomial-time NTMs:
\begin{lemma}
    \label{lem:graphs_setinstantiators_exist}
    Let $\varphi: \ordgraphs \to \IN$ be a pure ordered graph motif parameter, let $M$ be a polynomial-time NTM computing $\eval{\varphi}$ and let $G \in \ordgraphs$.
    Then there is a set-instantiator against $(M,G,\varphi)$.
\end{lemma}
\begin{proof}
    For now let $n \in \IN$ be undetermined, we will choose it accordingly later.
    Uniformly at random choose a function $\hat{\xi}: V(G) \to \{0, \ldots, n-1\}$ and let it induce a monotone function $\xi: V(G) \to \{0, \ldots, |V(G)| \cdot n - 1\}$ by mapping the $i$-th smallest vertex $v \in V(G)$ to $(i - 1)\cdot n + \hat{\xi}(v)$.
    We then set $j = |V(G)| \cdot n$.
    
    For an ordered graph $H \sqsubseteq G$, we denote by $\xi(H)$ the graph on the vertices $\{0, \ldots, |V(G)| \cdot n - 1\}$, where the edges are given as the images of the edges of $H$ under $\xi$.
    Any vertex that is not in the image of $\xi(V(H))$ is thus an isolated vertex.
    We are trying to find some $\xi$, s.t.\ for all $H \sqsubseteq G$, all accepting paths of the computation $\#\acc_{M^{\xi(H)}}(j)$ do not access the oracle for any edges incident to vertices in $\xi(V(G)) \setminus \xi(V(H))$.
    By the union bound, if we can show this happening with high probability for each of the finitely many individual $H \sqsubseteq G$, then there is such a $\xi$, so for now fix some $H \sqsubseteq G$.

    Now group the functions $\xi$ by their image of $V(H)$ and fix one such group $I$.
    Note that the image of the remaining vertices $V(G) \setminus V(H)$ under $\xi$ is still independently and uniformly distributed from $n$ possible choices each by the choice of $\hat{\xi}$.
    Consider the set $S_I$ of all vertices queried by all accepting paths of the computation $\#\acc_{M^{\xi(H)}}(j)$.
    Let $t_M(j)$ be the worst-case running-time of $M$ on input $j$.
    Each oracle query can query at most two vertices, there are at most $\varphi(H)$ many accepting paths and each accepting path can query the oracle at most $t_M(j)$ many times and as such the size $|S_I|$ is bounded by $2 \cdot \varphi(H) \cdot t_M(j)$, which is polynomial in $\log(j)$ due to $M$ being polynomial-time and $H$ being fixed.
    As a rough estimate using the Bernoulli inequality, this means that for a fraction of at least $\left(1 - \frac{|S_I|}{n}\right)^{|V(G)|} \geq 1 - \frac{|S_I| \cdot |V(G)|}{n}$ choices of $\hat{\xi}$ resulting in $\xi \in I$, none of the accepting paths queries any edges incident to vertices in $\xi(V(G)) \setminus \xi(V(H))$.
    By averaging across all possible $\hat{\xi}$ (no longer restricted to $I$), we see that a fraction of $1 - \frac{\polylog(j) \cdot |V(G)|}{n}$ choices have our desired property.

    Using the union bound over all $H \sqsubseteq G$, we get that a fraction of $1 - \frac{\polylog(j) \cdot |V(G)| \cdot 2^{|V(G)|}}{n}$ choices of $\hat{\xi}$ lead to the desired property for all $H \sqsubseteq G$ simultaneously.
    By choosing $n$ large enough (remember that $j$ and $n$ are polynomially related, so this is always possible), this guarantees the existence of our desired~$\xi$, s.t.\ for all $H \sqsubseteq G$, all accepting paths of the computation $\#\acc_{M^{\xi(H)}}(j)$ do not access the oracle for any edges incident to vertices in $\xi(V(G)) \setminus \xi(V(H))$.

    We can now construct our set-instantiator with $j$ chosen as above.
    The instantiation function for $H \sqsubseteq G$ is $\inst_\SI(H) = \xi(H)$.
    For any computation path $\tau \in \{0, 1\}^\star$ of a computation $\#\acc_{M^{\xi(H)}}(j)$, denote by $S_\tau \subseteq \{0, \ldots, j-1\}$ the set of vertices that are queried by the computation.
    We set
    \begin{align*}
        \perc_{\SI}(\tau) := \begin{cases}
            G[\xi^{-1}(S_\tau)] & \text{if $\tau$ is an accepting path of the computation $\#\acc_{M^{\xi(G)}}(j)$,}\\
            \top & \text{otherwise}
        \end{cases}
    \end{align*}
    Note that the preimage $\xi^{-1}(S_\tau)$ is defined to be simply the set $\{ v \in V(G) \mid \xi(v) \in S_\tau \}$, even if $S_\tau$ is not entirely contained in the image of $\xi$.

    We check the properties of a set-instantiator.
    Clearly $\varphi(\inst_\SI(H)) = \varphi(H)$ for all $H \sqsubseteq G$, since the two graphs are identical except isolated vertices and no pattern in $\varphi$ contains any isolated vertices.
    Further we need that for all $H \sqsubseteq G$ we have that $\tau \in \{0, 1\}^\star$ is an accepting path for the computation $\#\acc_{M^{\inst_\SI(H)}}(j)$ iff $\perc_\SI(\tau) \sqsubseteq H$.
    For this fix $H \sqsubseteq G$ and let $\tau$ be an accepting path for the computation $\#\acc_{M^{\inst_\SI(H)}}(j)$.
    By our choice of $\xi$ we know that no other vertices of $\xi(V(G)) \setminus \xi(V(H))$ are queried, or in other words, $\xi^{-1}(S_\tau) \subseteq V(H)$ and thus $\tau$ is also an accepting path for the computation $\#\acc_{M^{\inst_\SI(G)}}(j)$.
    This directly implies $\perc_{\SI}(\tau) = G[\xi^{-1}(S_\tau)] \sqsubseteq H$.
    On the other hand, let $\perc_{\SI}(\tau) \sqsubseteq H$, then $\perc_{\SI}(\tau) \neq \top$ and $\tau$ is an accepting path of the computation $\#\acc_{M^{\inst_\SI(G)}}(j)$, since $\perc_{\SI}(\tau) = G[\xi^{-1}(S_\tau)] \sqsubseteq H$, $\tau$ is also an accepting path of the computation $\#\acc_{M^{\inst_\SI(H)}}(j)$.
\end{proof}

If $H_1$ and $H_2$ are isomorphic, then it can still be the case that $\#\acc_{M^{\inst_\SI(H_1)}}(j) \neq \#\acc_{M^{\inst_\SI(H_2)}}(j)$, which we do not want to happen in a good function.
For example if $G$ is any graph containing some triangles, then we want the instantiation of every induced triangle of $G$ to have the same number of accepting paths.
We thus want to construct an instantiation function that no longer has this problem.
We achieve this by creating a very large set-instantiator and then restricting ourselves to some small part of the set-instantiator where we can enforce this property.

Ensuring this structure is achieved by usage of a Ramsey theorem.
For two graphs $H$ and $G$ we denote by $\binom{G}{H}$ the subset of induced subgraphs of $G$ that are isomorphic to $H$.
\begin{lemma}[Main Theorem in \cite{DBLP:journals/jct/Nesetril77}]
    \label{lem:nesetril}
    Let $F, G \in \ordgraphs$ and let $t \in \IN$ be fixed.
    Then there is a $H \in \ordgraphs$, such that for every $\Phi: \binom{H}{F} \to \{0, \ldots, t\}$, there is a $G_\Phi \sqsubseteq H$ with $G_\Phi \cong G$ with the property that $|\Phi(\binom{G_\Phi}{F})| = 1$.
\end{lemma}

Note that Lemma~\ref{lem:nesetril} is the reason why we need to generalize our result to ordered graphs first instead of being able to prove it directly for unordered graphs.
It is known that this type of Ramsey theorem does not hold for unordered graphs \cite[Theorem 5.1]{Nesetril1995}.

By repeatedly applying Lemma~\ref{lem:nesetril} we get
\begin{proposition}[Ramsey Theorem]
    \label{pro:graphs_ramseyposet}
    Let $G \in \ordgraphs$ and let $t \in \IN$ be fixed.
    Then there is an $H \in \ordgraphs$, such that for every $\Phi: \poset{H} \to \{0, \ldots, t\}$ there is a $G_\Phi \sqsubseteq H$ with $G_\Phi \cong G$ with the property, if $A, B \sqsubseteq G_\Phi$ with $A \cong B$, then $\Phi(A) = \Phi(B)$.
\end{proposition}

This theorem can now be used to ensure that for every polynomial-time NTM there must be some small set of inputs where the NTM behaves like a good function.

\begin{lemma}
    \label{lem:graphs_forceGood}
    Let $\varphi$ be a pure ordered graph motif parameter, let $M$ be any nondeterministic polynomial-time Turing machine computing $\eval{\varphi}$ and let $G \in \ordgraphs$.
    Then there is some $j \in \IN$, a function $\inst: \poset{G} \to \ordgraphssize{j}$ and a $\poset{G}$-good function $\Psi$ with $\#\acc_{M^{\inst(F)}}(j) = \Psi(F)$ and $\varphi(\inst(F)) = \varphi(F)$ for every $F \sqsubseteq G$.
\end{lemma}
\begin{proof}
    We invoke Proposition~\ref{pro:graphs_ramseyposet} for $G$ and $t := \max\{\varphi(H) \,:\, H \sqsubseteq G\}$ to obtain $\tilde{G} \in \ordgraphs$.
    We now want to color $\poset{\tilde{G}}$ according to the behaviour of $M$ on the elements of it, in particular how many accepting paths query each specific induced subgraph of $\tilde{G}$.
    Before we can do this we have to first embed all potential elements of $\poset{\tilde{G}}$ as oracles.
    In order to do this we use Lemma~\ref{lem:graphs_setinstantiators_exist} to construct a set-instantiator $\SI$ against $(M, \tilde{G}, \varphi)$ to obtain a $j \in \IN$, $\inst_\SI$ and $\perc_\SI$ as in Definition~\ref{def:graphs_setinstantiator}.
    
    For $H \in \poset{\tilde{G}}$ define $\Phi(H) := \#\acc_{M^{\inst_\SI(H)}}(j)$.
    Note that
    \begin{equation}\label{eq:PhiPerception}
        \Phi(H) \ = \ \bigl|\bigl\{\tau\in\{0,1\}^* \,: \, \perc_\SI(\tau)\sqsubseteq H\bigr\}\bigr| = \sum_{F \sqsubseteq H} \, \bigl|\big\{\tau\in\{0,1\}^* \,:\, \textsu{perc}_{\SI}(\tau)=F\big\}\bigr|.
    \end{equation}

    Using $\Phi$ as a coloring for Proposition~\ref{pro:graphs_ramseyposet}, we obtain $G_\Phi$.
    Note that $\Phi(H)$ now only depends on the isomorphism type $\isoclass{H}$ of $H$ for all $H \sqsubseteq G_\Phi$.
    By induction we see that the number
    \[
        \big|\big\{\tau\in\{0,1\}^* \, :\, \textsu{perc}_{\SI}(\tau)=H\big\}\big|
    \]
    also depends only on the isomorphism type of $H$, whenever $H \sqsubseteq G_\Phi$.
    Denote this number by $\alpha_{\isoclass{H}}$ and set $\Psi(\isoclass{H}) := \Phi(H)$.
    This is only well-defined for $H \sqsubseteq G_\Phi$, thus we consider the domain of $\Psi$ to be $\poset{\isoclass{G_\Phi}} = \poset{\isoclass{G}}$.
    This means that \eqref{eq:PhiPerception} simplifies:
    \begin{equation}\label{eq:PhiGood}
    \Psi(\isoclass{H}) \ = \ \sum_{\isoclass{F}} \ind{\isoclass{F}}{\isoclass{H}} \alpha_{\isoclass{F}}
    \end{equation}
    where the sum is over all isomorphism classes $\isoclass{F}$ of induced subgraphs $F \sqsubseteq H$.
    This implies that $\Phi$ is $\poset{G}$-good.

    It remains to define the function $\inst: \poset{G} \to \ordgraphssize{j}$.
    Let $H \sqsubseteq G$, then $\inst(H) := \inst_{\SI}(H_\Phi)$ for some $H_\Phi \sqsubseteq G_\Phi$ with $H \cong H_\Phi$.
    The property $\varphi(\inst(H)) = \varphi(H)$ now directly follows from $\SI$ being a set-instantiator.
\end{proof}

Note that Lemma~\ref{lem:graphs_forceGood} only guarantees the local behaviour of $M$ to be $\poset{G}$-good.
This leaves two possibilities for contradictions: Either the local behaviour is even $(\poset{G} \cap \ordgraphsnoniso)$-good, or the function grows too quickly for inputs with many isolated vertices.
This next Witness Theorem shows that in the first case we can always find an ordered graph that proves that $M$ does not compute the correct function using a simple linear algebra argument.

\begin{theorem}[The Witness Theorem]
\label{thm:graphs_findcounterexample}
Fix a bad $\varphi: \ordgraphsnoniso \to \IN$.
Then there exists $G \in \ordgraphsnoniso$ such that for
every $(\poset{G} \cap \ordgraphsnoniso)$-good function $\Psi:(\poset{G} \cap \ordgraphsnoniso) \to \IN$,
there exists $W \in \poset{G} \cap \ordgraphsnoniso$
with $\Psi(W) \neq \varphi(W)$.
\end{theorem}
\begin{proof}
    Choose $G$ to be the disjoint union of the graphs in $\supp(\varphi)$ and let $\ordgraphs' := \poset{G} \cap \ordgraphsnoniso$, keeping only one representative per isomorphism class.
    By Fact~\ref{fact: lin-indep-ind}, the functions $\ind{H}{\starpad}$ for $H\in \ordgraphs'$ are linearly independent.
    Thus the patterns and coefficients of any function $\Psi$ with support $\supp(\Psi) \subseteq \poset{G} \cap \ordgraphsnoniso$ are uniquely determined by the evaluations $\Psi$ on $\poset{G} \cap \ordgraphsnoniso$, and there exists a witness $W \in \poset{G} \cap \ordgraphsnoniso$ with $\Psi(W) \neq \varphi(W)$, since $\Psi$ is good while $\varphi$ is bad (thus, in particular, $\Psi$ and $\varphi$ are not the same function).
\end{proof}

With this last tool we are finally able to lead the existence of polynomial-time NTMs computing bad ordered graph motif parameters to a contradiction.

\begin{lemma}
    Let $M$ be any nondeterministic polynomial-time Turing machine computing $\eval{\varphi}$ for some bad pure ordered graph motif parameter $\varphi$.
    Then there is a $j \in \IN$ and $W \in \ordgraphssize{j}$ such that $\#\acc_{M^W}(j) \neq \varphi(W)$.
\end{lemma}
\begin{proof}
    We invoke the Witness Theorem~\ref{thm:graphs_findcounterexample} with our $\varphi$ 
    to obtain an $G_1 \in \ordgraphsnoniso$ as in the theorem.
    Construct $G_2$ such that
    \begin{align}
        \ind{H}{G_2} &> \max\big\{\varphi(F) \,:\, F \in \poset{G_1} \cap \ordgraphsnoniso \big\} && \text{for every $H \in \poset{G_1} \setminus \ordgraphsnoniso$} \label{eq:graphs_blowupNonR}\\
        \ind{H}{G_2} &= \ind{H}{G_1} && \text{for every $H \in \ordgraphsnoniso$} \label{eq:graphs_keepR}
    \end{align} by adding sufficiently many isolated vertices to $G_1$.
    For this it is sufficient to add $1 + \max\big\{\varphi(F) \,:\, F \in \poset{G_1} \cap \ordgraphsnoniso \big\}$ isolated vertices in between every pair of adjacent vertices in $G_1$, relative to $\leq_{G_1}$.

    We now use Lemma~\ref{lem:graphs_forceGood} to obtain $j \in \IN$, a function $\inst: \poset{G_2} \to \ordgraphssize{j}$ and a $\poset{G_2}$-good function $\Psi$ with $\#\acc_{M^{\inst(H)}}(j) = \Psi(H)$ and $\varphi(\inst(H)) = \varphi(H)$ for every $H \sqsubseteq G_2$.
    There are now two possibilities to find the required counterexample $W$:
    If there are no basis functions present in $\Psi$ with a pattern in $\poset{G_1} \setminus \ordgraphsnoniso$, we can invoke the Witness Theorem again, however if any of these basis functions are present we can construct a direct counterexample.

    We first handle the case where such a basis function is present, so let the coefficient of $\ind{H}{\starpad}$ in $\Psi$ be positive for some $H \in \poset{G_1} \setminus \ordgraphsnoniso$.
    Then
    \[
        \Psi(G_2) \geq \ind{H}{G_2} > \max\big\{\varphi(F) \,:\, F \in \poset{G_1} \cap \ordgraphsnoniso \big\} \geq \varphi(G_1) = \varphi(G_2)
    \]
    by \eqref{eq:graphs_blowupNonR} and \eqref{eq:graphs_keepR}, combined with the fact that $\varphi$ is pure,
    and our counterexample is $W := G_2$.

    Otherwise assume that no such basis function is present.
    Then $\tilde\Psi$ obtained from $\Psi$ by restricting to $\poset{G_1}$ (i.e.\ setting all coefficients of $\poset{G_2} \setminus \poset{G_1}$ to zero), is $(\poset{G_1} \cap \ordgraphsnoniso)$-good, which is a requirement for the second part of the Witness Theorem~\ref{thm:graphs_findcounterexample} (we already used it to obtain $G_1$).
    Furthermore $\tilde \Psi_{|\poset{G_1}} = \Psi_{|\poset{G_1}}$.
    We invoke the second part of the Witness Theorem~\ref{thm:graphs_findcounterexample} to find a point 
    $W \in \poset{G_1} \cap \ordgraphsnoniso$ with $\Psi (W) = \tilde\Psi (W) \neq \varphi(W)$ which is our counterexample.
    
    Independent of how we have obtained our $W$ we observe
    $\#\acc_{M^{\inst(W)}}(j) = \Psi(W) \neq \varphi(W) = \varphi(\inst(W))$,
    which completes the proof for $W' = \inst(W)$.
\end{proof}

This proves the separation 
in Theorem~\ref{thm:ordgraphs}, which together with Lemmas~\ref{lem:graphs_integercoeffs} and~\ref{lem:graphs_lower} finishes the proof of the theorem.

\section{Overview of Results for Relational Structures and Categories}

Induced subgraph counts can also be considered in a colored setting, where each vertex has a color and colored graphs are considered isomorphic if the colors of the vertices that are mapped to each other by the isomorphism have to agree.
Such colored patterns occur, e.g., within the proofs of some parameterized hardness results~\cite{DBLP:journals/algorithmica/DorflerRSW22,DBLP:conf/soda/Curticapean24}, as access to vertex-colors gives more versatility in the reduction.
It is natural to ask whether our main theorem generalizes from graphs to colored graphs and to other more general versions of graphs, such as hypergraphs of bounded edge-cardinality. We answer this question positively for general \emph{relational structures}; due to space limitations, this part appears as Appendix~\ref{sec:relational_structures}. The proof proceeds along the same lines as Theorem~\ref{thm:graphs}, our main theorem about graphs.

Generalizing further from relational structures, we consider the vastly more general setting of \emph{category theory} in Appendix~\ref{sec:category_general} and show a variant of Theorem~\ref{thm:graphs} for motif parameters defined over categories.
It turns out that the arguments from the main proof can be adapted for this purpose, but only after significant technical effort: While homomorphisms and embeddings naturally generalize from graphs to categories, nontrivial concepts from category theory are needed to formulate the proper requirements on categories that allow us to handle subobjects similarly to induced subgraphs.
Among other such requirements, we require \emph{well-powered} categories (where every subobject itself has only finitely many isomorphism types of subobjects) and \emph{proper factorization systems} (that decompose morphisms into compositions of epimorphisms and monomorphisms).
We then use the general category-theoretic framework to study combinatorial interpretations of counting problems related to finite vector spaces and to so-called \emph{parameter sets}, which play a role in enumerative combinatorics.

\bibliography{lit}

\appendix
\section{Omitted proofs}

\begin{proof}[Proof of Fact~\ref{fact: lin-indep-ind}]
We show the fact for unordered graphs; the version for ordered graphs is shown along the same lines.
Let $\{G_1,\ldots,G_t\}$ for $t \in \nn$ be an enumeration of $\graphs'$ such that $|V(G_i)| < |V(G_j)|$ implies $i < j$.
Consider the matrix $A$ whose rows and columns are indexed by $G_1,\ldots,G_t$, with entries $\indsub{G_i}{G_j}$.

We show that this matrix is lower triangular with all diagonal entries equal to $1$, implying full rank and thus linear independence of the functions $f_H$:
Consider indices $i,j$ with $i\geq j$: 
First, we have $|V(G_i)| \geq |V(G_j)|$ by the choice of enumeration. If $A(i,j) \neq 0$ holds, it must follow that $|V(G_i)| = |V(G_j)|$, as otherwise $G_i$ has too many vertices to be an induced subgraph of $G_j$. But then $A(i,j) \neq 0$ only for $j=i$. In this case, $A(i,j)=1$.
\end{proof}

\begin{proof}[Proof of Lemma~\ref{lem:graphs_lower}]
    Since $\varphi$ is a positive integer linear combination of finitely many induced subgraph numbers, and $\sharpPdottedcircle$ is closed under such linear combinations, it suffices to consider $\varphi(G) = \indsub{H}{G}$ for $H \in \graphs$.
    We show the statement for unordered graphs, the proof for ordered graphs is analogous.

    A nondeterministic Turing machine for $\eval{\varphi}$ proceeds as follows, for hardcoded $H$ with $|V(H)|=k$:
    Given $(x,O)$ encoding a graph $G$, nondeterministically guess strings $x_1 < \ldots <  x_k < x$ of length $|x|$, where $<$ denotes the lexicographic ordering on strings.
    Then, for every $\{i,j\} \in \binom{k}{2}$, write $x_i x_j$ on the oracle tape and query $O$ to obtain a bit $a_{i,j}$.
    Finally, test by brute-force whether the graph with adjacency matrix given by $a_{i,j}$ is isomorphic to $H$; this only requires constant time, as $H$ is fixed.
    If the test succeeds, then this branch accepts, otherwise it rejects.
    This nondeterministic Turing machine clearly runs in polynomial-time on every branch, and the number of accepting computation paths equals the number of vertex-subsets $S$ of $G$ such that $G[S]$  is isomorphic to $H$.  
\end{proof}

\section{Relational Structures}
\label{sec:relational_structures}

First, we need to transfer some terminology from graphs to relational structures and introduce some additional notions that are specific for relational structures. Let $k_1, \ldots, k_\ell > 0$. 
A finite relational structure of \emph{type} $\Delta := (k_1, \ldots, k_\ell)$ consists of a finite universe $V$ of individuals (which we call vertices, in analogy to our terminology for graphs) and relations $R_1, \ldots, R_\ell$ with $R_i \subseteq \binom{V}{k_i}$.
Denote the class of all finite relational structures of type $\Delta$ as $\rel{\Delta}$.
From now on, we will refer to finite relational structures simply as relational structures.
For any relational structure $A = (V, R_1, \ldots, R_\ell)$, we write $V(A) := V$ and $R_i(A) = R_i$ for $i=1\ldots \ell$.

A relational structure $A$ is an \emph{induced substructure} of $B$, written as $A \sqsubseteq B$, if $V(A) \subseteq V(B)$ and $R_i(A) = R_i(B) \cap \binom{V(A)}{k_i}$ for all $1 \leq i \leq \ell$.
Furthermore $A$ is a \emph{substructure} of $B$, written as $A \subseteq$ B, if $V(A) \subseteq V(B)$ and $R_i(A) \subseteq R_i(B) \cap \binom{V(A)}{k_i}$ for all $1 \leq i \leq \ell$.
A function $f: V(A) \to V(B)$ is called a \emph{strong embedding} if $f$ is injective and $X \in R_i(A) \Leftrightarrow f(Y) \in R_i(B)$ for all $1 \leq i \leq \ell$ and $X \subseteq V(A)$.

A relational structure $A$ is called \defn{irreducible} if for all different $u, v \in V(A)$ there is some $1\leq i \leq \ell$ and $X \in R_i(A)$ with $u, v \in X$.

A vertex $v \in V(A)$ is called $P$-padding for $P = (p_1, \ldots, p_\ell) \in \{0, 1\}^\ell$ if for all $p_i = 0$ we have $\{X \in R_i(A) : v \in X\} = \emptyset$ and for all $p_i = 1$ we have $\{X \in R_i(A) : v \in X\} = \{ X \in \binom{V(A)}{k_i}  : v \in X\}$, i.e.,\ $v$ is not in a relation with any other vertices for $p_i = 0$ and in all possible relations with other vertices for $p_i = 1$.
A relational structure $A$ is now called $P$-pure if it does not contain any $P$-padding vertices.

A subset $F \subseteq \rel{\Delta}$ is called \defn{upwards closed} if for every $A \in F$ and all $B \in \rel{\Delta}$ with $V(A) = V(B)$ and $A \subseteq B$ we have $B \in F$.
Let $F \subseteq \rel{\Delta}$ be a potentially empty upwards closed set of $P$-pure irreducible relational structures of type $\Delta$.
Now let $U = \forb{F}$ be the set of all relational structures of type $\Delta$ that do not contain any substructure isomorphic to any element in $F$.
This is equivalent to $U$ not containing any induced substructures isomorphic to elements in $F$, due to $F$ being upwards closed.
Restricting $U$ to only the $P$-pure relational structures yields $\ppure{U}$.

As was the case for graphs, the Ramsey theorem only holds for an ordered version of relational structures.
We say $\ordrel{\Delta}$ is the class of all ordered relational structures of type $\Delta$.
Ordered relational structures are defined, as in the graph case, by incorporating a linear order $\leq_V$ on the universe into the structure and extending notions of homomorphisms, (strong) embeddings, and isomorphisms to respect the ordering, which also yields notions of (induced) substructures for ordered structures.
For a relational structure $A$ and $X\subseteq V(A)$, write $A[X]$ for the substructure induced by $X$.
Parallel to Definition~\ref{def:graph-motif-parameter}, define (ordered) motif parameters $\varphi$ to be linear combinations of induced substructure counts $\indsub{A}{\starpad}$ for fixed (ordered) relational structures $A$.
We call $\varphi$ $P$-pure if all its patterns are $P$-pure. For a class of relational structures $U$, we speak of an $U$ motif parameter if its domain is restricted to $U$.

Relational structures $A$ with $V(A) = \{0, \ldots, j-1\}$ are encoded as oracles $(x, O)$ by choosing $x = j$, and by querying $O$ with $(i, (v_1, \ldots, v_{k_i}))$, a Turing machine can ask whether $\{v_1, \ldots, v_{k_i}\} \in R_i(A)$.

We can now state our main theorem, in analogy to Theorem~\ref{thm:ordgraphs}.
\begin{theorem}
\label{thm:ordrelstructures}
Let $\varphi: U \to \IN$ be a $P$-pure ordered $U$ motif parameter.
Then $\eval{\varphi} \in \promise\sharpPdottedcircle$ iff all coefficients of $\varphi$ in the $\indsub{A_i}{\starpad}$ basis are nonnegative integers.
\end{theorem}

We remind that $\eval{\varphi}$ is a promise problem.
As such the NTM only has to count correctly whenever the pair $(x, O)$ correctly encodes an ordered relational structure from the domain of $\varphi$, in this case $U$.
This for example implies that for $\Delta = (2)$, $P = (0)$ and $F = \{K_3\}$ (the clique on 3 vertices), the counter example oracle input for each NTM will itself be a graph without any triangles.

The proof of Theorem~\ref{thm:ordrelstructures} is almost identical to that of Theorem~\ref{thm:ordgraphs}, so we will only point out some specific modifications required in the proof:
\begin{itemize}
    \item The Ramsey theorem: Lemma~\ref{lem:nesetril} was shown in~\cite{DBLP:journals/jct/Nesetril77} not only for the class of all finite graphs, but more generally for the class of all relational structures defined via a set of irreducible forbidden substructures, such as $U$.
    \item The value of a $P$-pure ordered motif parameter $\varphi$ does not change by adding or removing $P$-padding vertices, since any $P$-padding vertex $v$ in an ordered relational structure stays a $P$-padding vertex in any induced substructure containing $v$.
    Anywhere we were filling up a graph with isolated vertices, we instead add $P$-padding vertices. (Note that isolated vertices are simply $(0)$-padding vertices.)
    \item By the same argument, it is impossible to leave $U$ by adding $P$-padding vertices to a relational structure that was in $U$.
    \item Due to the changes of the way we store an ordered relational structure as an oracle, we will be showing the existence of set-instantiators explicitly again.
\end{itemize}

The notion of a set-instantiator is generalized as follows. We write $U_{=j}$ for the ordered relations in $U$ whose universes are exactly the vertices $\{0, \ldots, j-1\}$ with the natural order on them. 
\begin{definition}[Set-instantiator against $(M,A,\varphi)$]
\label{def:rel_setinstantiator}
Let  $M$  be a nondeterministic Turing machine,
let $A \in U_{=j}$ and
let $\varphi$ be an ordered $P$-pure motif parameter.

Let  $\top$  be a symbolic top element above the induced substructure poset $\poset{A}$, i.e., $\top \not\sqsubseteq B$ for all induced substructures $B \sqsubseteq A$.
A  \defn{set-instantiator}  $\SI$  is a triple of

$\circ$ \ some $j \in \IN$, 

$\circ$ \ an instantiation function  $\inst_\SI : \poset{A} \to U_{=j}$,   and

$\circ$ \  a perception function  $\perc_\SI : \{0,1\}^* \to \poset{A} \cup \{\top\}$,

\smallskip
\noindent
such that the following property holds for all induced substructures $B \sqsubseteq A$:

\smallskip
$\bu$ \ $\tau\in\{0,1\}^*$ is an accepting path for the computation  $\#\acc_{M^{\inst_\SI(B)}}(j)$
if and only if $\perc_\SI(\tau) \sqsubseteq B$, and 

$\bu$ \ $\varphi(\inst_\SI(B))) = \varphi(B)$.
\end{definition}
\begin{lemma}
    \label{lem:rel_setinstantiators_exist}
    Let $\varphi: U \to \IN$ be a $P$-pure ordered $U$ motif parameter, let $M$ be a polynomial-time NTM computing $\eval{\varphi}$ and let $A \in U$.
    Then there is a set-instantiator against $(M,A,\varphi)$.
\end{lemma}
\begin{proof}
    For now let $n \in \IN$ be undetermined, we will choose it accordingly later.
    Uniformly at random choose a function $\hat{\xi}: V(A) \to \{0, \ldots, n-1\}$ and let it induce a monotone function $\xi: V(A) \to \{0, \ldots, |V(A)| \cdot n - 1\}$ by mapping the $i$-th smallest vertex $v \in V(A)$ to $(i - 1)\cdot n + \hat{\xi}(v)$.
    We then set $j = |V(A)| \cdot n$.
    
    For an ordered relational structure $B \sqsubseteq A$, we denote by $\xi(B)$ the relational structure on the vertices $\{0, \ldots, |V(G)| \cdot n - 1\}$, where the relations are given as the images of the relations of $B$ under $\xi$.
    Furthermore, for every vertex $v$ that is not in the image of $\xi(V(A))$, make $v$ a $P$-padding vertex.
    We are trying to find some $\xi$, s.t.\ for all $B \sqsubseteq A$, all accepting paths of the computation $\#\acc_{M^{\xi(B)}}(j)$ do not access the oracle for any relations incident to vertices in $\xi(V(A)) \setminus \xi(V(B))$.
    By the union bound, if we can show this happening with high probability for each of the finitely many individual $B \sqsubseteq A$, then there is such a $\xi$, so for now fix some $B \sqsubseteq A$.

    Now group the functions $\xi$ by their image of $V(B)$ and fix one such group $I$.
    Note that the image of the remaining vertices $V(A) \setminus V(B)$ under $\xi$ is still independently and uniformly distributed from $n$ possible choices each by the choice of $\hat{\xi}$.
    Consider the set $S_I$ of all vertices queried by all accepting paths of the computation $\#\acc_{M^{\xi(B)}}(j)$.
    Let $t_M(j)$ is the worst-case running-time of $M$ on input $j$.
    Each oracle query queries one of the relations at a time and thus at most $\max(k_1, \ldots, k_\ell)$ vertices, there are at most $\varphi(B)$ many accepting paths and each accepting path can query the oracle at most $t_M(j)$ many times and as such the size $|S_I|$ is bounded by $\max(k_1, \ldots, k_\ell) \cdot \varphi(B) \cdot t_M(j)$, which is polynomial in $\log(j)$ due to $M$ being polynomial-time and $B$ and $k_1, \ldots, k_\ell$ being fixed.
    As a rough estimate using the Bernoulli inequality and $|V(A) \setminus V(B)| \leq |V(A)|$, this means that for a fraction of at least $\left(1 - \frac{|S_I|}{n}\right)^{|V(A)|} \geq 1 - \frac{|S_I| \cdot |V(A)|}{n}$ choices of $\hat{\xi}$ resulting in $\xi \in I$, none of the accepting paths queries any edges incident to vertices in $\xi(V(A)) \setminus \xi(V(B))$.
    By averaging across all possible $\hat{\xi}$ (no longer restricted to $I$), we see that a fraction of $1 - \frac{\polylog(j) \cdot |V(A)|}{n}$ choices have our desired property.

    Using the union bound over all $B \sqsubseteq A$, we get that a fraction of $1 - |\poset{A}| \cdot \frac{\polylog(j) \cdot |V(A)|}{n}$ choices of $\hat{\xi}$ lead to the desired property for all $B \sqsubseteq A$ simultaneously.
    By choosing $n$ large enough (remember that $j$ and $n$ are polynomially related, so this is always possible), this guarantees the existence of our desired~$\xi$, s.t.\ for all $B \sqsubseteq A$, all accepting paths of the computation $\#\acc_{M^{\xi(B)}}(j)$ do not access the oracle for any edges incident to vertices in $\xi(V(A)) \setminus \xi(V(B))$.

    We can now construct our set-instantiator with $j$ chosen as above.
    The instantiation function for $B \sqsubseteq A$ is
    \begin{align*}
        \inst_\SI(B) &:= \xi(B)\,.
    \end{align*}
    For any computation path $\tau \in \{0, 1\}^\star$ of a computation $\#\acc_{M^{\xi(B)}}(j)$, denote by $S_\tau \subseteq \{0, \ldots, j-1\}$ the set of vertices that are queried by the computation.
    We set
    \begin{align*}
        \perc_{\SI}(\tau) := \begin{cases}
            A[\xi^{-1}(S_\tau)] & \text{if $\tau$ is an accepting path of the computation $\#\acc_{M^{\xi(A)}}(j)$,}\\
            \top & \text{otherwise}
        \end{cases}
    \end{align*}

    It remains to check the properties of a set-instantiator.
    Clearly $\varphi(\inst_\SI(B)) = \varphi(B)$ for all $B\sqsubseteq G$, since the two relational structures are identical except $P$-padding vertices and no pattern in $\varphi$ contains any $P$-padding vertices.
    Further we need that for all $B \sqsubseteq A$ we have that $\tau \in \{0, 1\}^\star$ is an accepting path for the computation $\#\acc_{M^{\inst_\SI(B)}}(j)$ iff $\perc_\SI(\tau) \sqsubseteq B$.
    For this fix $B \sqsubseteq A$ and let $\tau$ be an accepting path for the computation $\#\acc_{M^{\inst_\SI(B)}}(j)$.
    By our choice of $\xi$ we know that no other vertices of $\xi(V(A)) \setminus \xi(V(B))$ are queried, or in other words, $\xi^{-1}(S_\tau) \subseteq V(B)$ and thus $\tau$ is also an accepting path for the computation $\#\acc_{M^{\inst_\SI(A)}}(j)$.
    This directly implies $\perc_{\SI}(\tau) = A[\xi^{-1}(S_\tau)] = \sqsubseteq B$.
    On the other hand, let $\perc_{\SI}(\tau) \sqsubseteq B$, then $\perc_{\SI}(\tau) \neq \top$ and $\tau$ is an accepting path of the computation $\#\acc_{M^{\inst_\SI(A)}}(j)$, since $\perc_{\SI}(\tau) = A[\xi^{-1}(S_\tau)] \sqsubseteq B$, $\tau$ is also an accepting path of the computation $\#\acc_{M^{\inst_\SI(B)}}(j)$.
\end{proof}

\subsection{Directed Relational Structures}
So far for all of our graphs and relational structures we have assumed that they are undirected and have no self-loops or relations with duplicate vertices.

We study two further generalizations to ordered and unordered relational structures $A$, where each relation $R_i$ can now be one of following different variants:
\begin{enumerate}
    \item Set relations, i.e.\ $R_i \subseteq \binom{V}{k_i}$ which is the kind of relational structures we have studied so far.
    \item Multiset relations, i.e.\ $R_i$ is a set of multiset of vertices from $V$ with $k_i$ total (not necessarily distinct) elements each.
    These for example would allow us to express undirected self-loops in the graph setting.
    \item List relations without repetitions, i.e.\ $R_i$ is a tuple of $k_i$ vertices from $V$, but without repetitions.
    Directed graphs without self-loops fall under this category.
    \item List relations with repetitions, i.e.\ $R_i$ is a tuple of $k_i$ elements from $V$ with repetitions.
    Directed graphs with self-loops fall under this category.
\end{enumerate}
We can even allow different relations to be of different variants.
We call such an (ordered) relational structure a mixed (ordered) relational structure.
A \defn{mixed type} $\tilde{\Delta}$ is now a type $\Delta$, combined with the information about which variant each relation uses.
The class $\mordrel{\tilde{\Delta}}$ is then the class of all mixed ordered relational structures of type $\tilde{\Delta}$.
The induced substructures of a mixed (ordered) relational structure are now defined in the natural way and homomorphisms and (strong) embeddings preserve the order and multiplicities of the mixed relations.
We call a mixed (ordered) relational structure $A$ irreducible if for every different $u, v \in V(A)$ there is some $1 \leq i \leq \ell$ and $X \in R_i(A)$, such that $X$ contains both $u$ and $v$ (as a set, multiset or list, depending on the variant of $R_i$).
For a vertex $v$ to be a $P$-padding vertex the relation $R_i$ needs to contain all possible entries including $v$ if $p_i = 1$ and for $p_i = 0$ it has to contain none of the entries including $v$.
As usual $A$ is called $P$-pure if it contains no $P$-padding vertices.

In an ordered relational structure, each of these types of relations can be expressed by using one or more set relations in the following way.
\begin{description}
    \item[Multiset relation:]
        Replace $R_i$ by one set relation $(R_i)_C$ of arity $|C|$ for every composition $C$ of $k_i$.
        Let $e$ be a hyperedge in $R_i$ with vertices $v_1, \ldots, v_\nu$ with multiplicities $\lambda_1, \ldots, \lambda_\nu$.
        W.l.o.g.\ assume that $v_1 \leq v_2 \leq \ldots \leq v_\nu$ according to the linear order on $V$.
        Then $\{v_1, \ldots, v_\nu\}$ is added to $(R_i)_{(\lambda_1, \ldots, \lambda_\nu)}$.
    \item[List relation without repetition:]
        Replace $R_i$ by one set relation $(R_i)_\prec$ of arity $k_i$ for every possible total order $\prec$ of $\{1, \ldots, k_i\}$.
        Let $(v_1, \ldots, v_{k_i}) \in R_i$.
        Define the order $\prec_{(v_1, \ldots, v_{k_i})}$ on $\{1, \ldots, k_i\}$ via $i \prec j := v_i \leq v_j$ and add $\{v_1, \ldots, v_{k_i}\}$ to $(R_i)_{\prec_{(v_1, \ldots, v_{k_i})}}$.
    \item[List relation with repetition:]
        Same as above, however $\prec$ is now a total preorder, i.e.\ it does not have to be anti-symmetric.
        This also implies that not all of the added relations have arity exactly $k_i$, but rather arity $\leq k_i$, similarly to the multiset relation case.
\end{description}
Denote the resulting type of the ordered relational structures as $\Delta'$.
Note that $\Delta'$ only depends on $\tilde{\Delta}$.
We call the function applying the above conversion rules $\convert: \mordrel{\tilde{\Delta}} \to \ordrel{\Delta'}$.
\begin{fact}
    \label{fact:mixed_bijective}
    This construction is bijective, i.e.\ the conversion mapping $\convert$ is bijective.
    Furthermore, both $\convert$ and $\convert^{-1}$ are efficiently computable in the sense that it is possible, in polynomial time in $\log(|V(A)|)$, to compute the image (or preimage) of any possible entry of the relations under $\convert$.
\end{fact}
\begin{fact}
    \label{fact:mixed_commute}
    Moreover $\convert$ commutes with (strong) embeddings and thus induced subgraphs and isomorphisms.
    In particular this directly implies that $\ind{B}{A} = \ind{\convert(B)}{\convert(A)}$.
\end{fact}

These two facts also imply a further generalization of the Ramsey theorem in \cite{DBLP:journals/jct/Nesetril77} to mixed ordered relational structures by applying $\convert$ before using the Ramsey theorem on ordered relational structures and then applying $\convert^{-1}$ to get convert the resulting relational structure back to a mixed ordered relational structure.
However we do not use this variant of the Ramsey theorem here and instead do a direct reduction.

\begin{theorem}
    \label{thm:ordmixedrelstructures}
    Let $\tilde{\Delta}$ be a mixed type with $\ell$ relations, let $P \in \{0, 1\}^\ell$ and let $F \subseteq \mordrel{\tilde{\Delta}}$ be an upwards closed subset of $P$-pure irreducible mixed ordered relational structures.
    
    Let $\varphi: \forb{F} \to \IN$ be a $P$-pure ordered $\forb{F}$ motif parameter, i.e.,
    $\varphi(B)= \sum_{i=1}^s \alpha_i \cdot \ind{A_i}{B}$
    for pairwise non-isomorphic $P$-pure $A_1,\ldots,A_s$ and coefficients $\alpha_1 ,\ldots,\alpha_s \in \qqq$.
    Then the evaluation problem $\eval{\varphi}$ is in $\promise\sharpPdottedcircle$ iff all coefficients $\alpha_1 ,\ldots,\alpha_s \in \qqq$ are nonnegative integers.
\end{theorem}
\begin{proof}
    Let $F' = \{\convert(A) : A \in F\}$ and let $P'$ be obtained from $P$ by setting $p'_{j_i} = p_i$ where $j_i$ ranges over all the relations added by the conversion from relation $i$.
    Then $F'$ is upwards closed and only contains $P'$-pure ordered relational structures.
    We now define $\varphi': \forb{F'} \to \IN$ as $\varphi'(B')= \sum_{i=1}^s \alpha_i \cdot \ind{\convert(A_i)}{B'}$.
    All the $\convert(A_i)$ are pairwise non-isomorphic and $P'$-pure.
    By Theorem~\ref{thm:ordrelstructures} we know that $\eval{\varphi'} \in \promise\sharpPdottedcircle$ iff $\alpha_1, \ldots, \alpha_s$ are nonnegative integers.
    Given an oracle encoding some $B \in \forb{F}$ we can now in polynomial time simulate an oracle encoding $\convert(B) \in \forb{F'}$ by using Fact~\ref{fact:mixed_bijective}.
    Moreover by Fact~\ref{fact:mixed_commute} $\varphi(B) = \varphi'(\convert(B))$, thus if $\eval{\varphi'} \in \promise\sharpPdottedcircle$, then also $\eval{\varphi} \in \promise\sharpPdottedcircle$.

    Since $\convert$ is bijective and efficiently computable in both directions this argument also reverses and proves $\eval{\varphi'} \in \promise\sharpPdottedcircle \Leftrightarrow \eval{\varphi} \in \promise\sharpPdottedcircle$.
\end{proof}

In the same way that we have proven Theorem~\ref{thm:graphs} from Theorem~\ref{thm:ordgraphs}, we also get a variant with unordered mixed relational structures.

\begin{theorem}
    \label{thm:mixedrelstructures}
    Let $\tilde{\Delta}$ be a mixed type with $\ell$ relations, let $P \in \{0, 1\}^\ell$ and let $F \subseteq \mrel{\tilde{\Delta}}$ be an upwards closed subset of $P$-pure irreducible mixed relational structures.
    
    Let $\varphi: \forb{F} \to \IN$ be a $P$-pure $\forb{F}$ motif parameter, i.e.,
    $\varphi(B)= \sum_{i=1}^s \alpha_i \cdot \ind{A_i}{B}$
    for pairwise non-isomorphic $P$-pure $A_1,\ldots,A_s$ and coefficients $\alpha_1 ,\ldots,\alpha_s \in \qqq$.
    Then the evaluation problem $\eval{\varphi}$ is in $\promise\sharpPdottedcircle$ iff all coefficients $\alpha_1 ,\ldots,\alpha_s \in \qqq$ are nonnegative integers.
\end{theorem}
\begin{proof}
    We observe 
    \[
        \ind{(V_i, R_{i,1}, \ldots, R_{i, \ell})}{(V, R_1, \ldots, R_\ell)} = \sum_{\leq_{V_i}}\indsub{(V_i, R_{i,1}, \ldots, R_{i, \ell}, \leq_{V_i})}{(V, R_1, \ldots, R_\ell, \leq_V)}
    \]
    where $\leq_V$ is any arbitrary linear ordering of $V$ and $\leq_{V_i}$ sums over all linear orders of $V_i$ that result in non-isomorphic ordered relational structures.
    Further let $F' \subseteq \mordrel{\Delta}$ be obtained from the elements of $F$ combined with all possible linear orders of their vertices.
    Note that all elements of $F'$ are irreducible and $P$-pure and $F'$ is still upwards closed, we can thus invoke Theorem~\ref{thm:ordmixedrelstructures}.

    The same argument as for Theorem~\ref{thm:graphs} finishes the proof.
\end{proof}

\subsection{Applications}
For colored graphs, a special case of relational structures, we get a particularly clean dichotomy, here we do not rely on the patterns not containing any isolated vertices, we only require that there is some color that can be used as a color in the target graph that does not appear in any of the pattern graphs.
Let $\colgraphs{c+1}$ to be the set of all colored graphs on $c+1$ colors.
We rephrase $\colgraphs{c+1}$ as relational structures with $\Delta = (2, \underbrace{1, \ldots, 1}_{c \text{\ times}})$, by adding a unary relation for every color, except the color $0$.
By choosing $P=(0, \ldots, 0)$ we set $P$-padding vertices to be exactly those vertices using the color $0$.
A colored graph is thus $P$-pure iff it contains no vertices with color $0$.
Furthermore we set $F$ to be the set of relational structures on a single vertex, such that at least two different of the unary relations are set, these correspond to vertices that are assigned more than one color, turning $\forb{F}$ into exactly the set of colored graphs where each vertex has exactly one color, either one of $1, \ldots, c$ or the color $0$ by absence of any other color.

By using Theorem~\ref{thm:mixedrelstructures} we thus get
\begin{theorem}
    \label{thm:colored_graphs}
    Let $\varphi: \colgraphs{c+1} \to \IN$ be a colored graph motif parameter, i.e.,
$\varphi(G)= \sum_{i=1}^s \alpha_i \cdot \indsub{H_i}{G}$
for pairwise non-isomorphic graphs $H_1,\ldots,H_s$ using the colors $\{1, \ldots, c\}$ and coefficients $\alpha_1 ,\ldots,\alpha_s \in \qqq$.
Then the evaluation problem $\eval{\varphi}$ is in $\promise\sharpPdottedcircle$ iff all coefficients $\alpha_1 ,\ldots,\alpha_s \in \qqq$ are nonnegative integers.
\end{theorem}

Note that while we restrict the number of colors within the theorem itself, the same result holds true if the number of colors is allowed to grow with $\varphi$, by simply choosing the number of colors for the target graph to be one bigger than the fixed number of colors used in the patterns of $\varphi$.

\section{Categorical Generalization}
\label{sec:category_general}

In the proofs of \S\ref{sec:graphs} we didn't use particularly many properties of ordered graphs to show our lower bound, except some structure of (induced) subobjects, a Ramsey property and the existence of set-instantiators.
The goal of this section is now to generalize the proof and make it more explicit what properties in detail are needed for our proof to work.

\subsection{Basic Category Theory}
\label{sec:category_basics}
We follow the definitions of~\cite{riehl2017category} and refer the reader to it should they want more than just these basic definitions.
A category $C$ consists of two collections:
\begin{itemize}
    \item A collection of \defn{objects}.
    \item A collection of \defn{morphisms} between objects.
\end{itemize}
such that:
\begin{itemize}
    \item Every morphism $f$ has a specified \defn{domain} and a \defn{codomain}.
    We write $f: a \to b$ if $f$ has domain $a$ and codomain $b$.
    Additionally we write $\domain f = a$ and $\codomain f = b$.
    \item For every object $a$ there is an \defn{identity morphism} $\identity_a: a \to a$.
    \item For every two morphisms $f: a \to b$ and $g: b \to c$, there is a specific \defn{composite morphism} $gf: a \to c$.
    \item For every morphism $f: a \to b$ both $f\identity_a$ and $\identity_bf$ are identical to $f$, i.e.\ the identity morphisms are the neutral elements of composition.
    \item For every triple of morphisms $f: a \to b$, $g: b \to c$ and $h: c \to d$ we have that $h(gf)$ and $(hg)f$ are identical, i.e.\ composition is associative.
\end{itemize}

A morphism $f: b \to c$ is called
\begin{itemize}
    \item a \defn{monomorphism}, if for any morphisms $g, h: a \to b$ with $fg = fh$ we already have $g = h$.
    We also say that $f$ is monic or a mono.
    \item an \defn{epimorphism}, if for any morphisms $g, h: c \to d$ with $gf = hf$ we already have $g = h$.
    We also say that $f$ is epic or an epi.
    \item an \defn{isomorphism}, if there is a morphism $g: c \to b$ with $fg = \identity_c$ and $gf = \identity_b$.
    In this case we call $b$ and $c$ \defn{isomorphic}, also written as $b \cong c$.
\end{itemize}

If for every two objects $a, b$ the collection of all morphism from $a$ to $b$ is a set\footnote{We don't want to to get into the set versus class discussion too much here. While some of our categories have a class worth of objects, they are essentially small, meaning there is only a set worth of non-isomorphic objects. Further all of our categories are locally small.}, then $C$ is called \defn{locally small}.
If this collection is even finite, then it is called \defn{locally finite}.

We call some category $D$ a \defn{subcategory} of $C$ if the objects of $D$ are a subcollection of the objects of $C$ and the morphisms of $D$ are a subcollection of the morphisms of $C$.
Note that this includes all relevant identity and composite morphisms in order for $D$ to be a category.
We say $D$ is a \defn{full subcategory} of $C$ if, additionally, for every two objects $a, b$ in $D$, it contains all morphisms with domain $a$ and codomain $b$ from $C$.

The last concept we need is the concept of finite coproducts.
For objects $a_1, \ldots, a_k$ we say that an object $b$, together with morphisms $\iota_1, \ldots, \iota_k$ where $\iota_i: a_i \to b$, called the \defn{coproduct injections}\footnote{These coproduct injections do not have to be injective, nor mono}, is the \defn{coproduct} of $a_1, \ldots, a_k$, if for all morphisms $f_1, \ldots, f_k$ with $f_i: a_i \to c$ there is a unique morphism $f:b \to c$ such that the following diagram commutes:
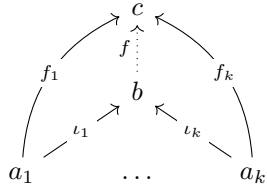
\begin{center}
    \begin{tikzcd}
                                                                                    & c                                                       &                                                                                  \\
                                                                                    & b \arrow[u, "f", dotted] &                                                                                  \\
    a_1 \arrow[ru, "\iota_1" description] \arrow[ruu, "f_1" description, bend left] & \ldots                                                  & a_k \arrow[lu, "\iota_k" description] \arrow[luu, "f_k" description, bend right]
    \end{tikzcd}
\end{center}
The object $b$ is unique up to isomorphism, so we write $\bigsqcup_{i=1}^k a_i$ for it.
We also write $\bigsqcup_{i=1}^kf_i$ for the morphism $f$.

\subsection{Counting in Categories}

Our objects of consideration are now objects of some category $C$.
To specify the subobjects we are counting, we designate a class of morphisms from $C$:
\begin{definition}
    \label{def:Msub}
    Let $C$ be a locally small category and let $\cM$ be some class of morphisms in $C$ that contains all isomorphisms and is closed under composition.
    For any objects $a, b$ we call morphisms $a \to b$ from $\cM$ the $\cM$-subobjects of $a$ under $b$.
    We identify morphisms $f, f': a \to b$ if there is an isomorphism $g: a \to a$ with $f' = gf$.
    The set of all $\cM$-subobjects of $a$ under $b$ is denoted by $\Msub{a}{b}$, while the class of all $\cM$-subobjects under $b$ is denoted by $\poset{b}$.
    If all such $\poset{b}$ are small (i.e.\ a set), then we call $C$ $\cM$-well-powered and if they are all finite $C$ is called finitely $\cM$-well-powered.
    For notational convenience, we write $\domain \poset{b} := \{\domain f \,:\, f \in \poset{b}\}$.

    Let $f: a \to b$ and $f': a' \to b$ be $\cM$-subobjects under $b$.
    We write $f' \prec f$ if there is some morphism $g$ with $f' = gf$.
\end{definition}

The class $\cM$ in here should be thought of as some class of special monomorphisms.
In the example of undirected graphs $\cM$ is the set of strong graph embeddings, i.e.\ injective graph homomorphisms that preserve both edges and non-edges.
The poset $\poset{G}$ is then the set of all induced subgraphs of $G$ (represented via their graph embeddings), while $\domain \poset{G}$ contains all graphs with isomorphism types that are among the induced subgraphs of $G$.

Combined with a second class of morphisms $\cE$, we obtain a factorization system:
\begin{definition}
    \label{def:factorization_system}
    Let $C$ be a category and let $\cE$ and $\cM$ be some classes of morphisms in $C$.
    We call ($\cE$, $\cM$) a factorization system if they fulfil the following properties:
    \begin{enumerate}
        \item Every morphism $f: a \to b$ factors as $f = me$ with $e \in \cE$ and $m \in \cM$.
        This factorization is unique, up to a unique isomorphism, i.e.\ if there is some other factorization $f = m'e'$ with $e' \in \cE$ and $m' \in \cM$, then there is a unique isomorphism $g$ such that the following diagramm commutes:

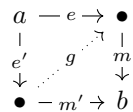
        \begin{center}
            \begin{tikzcd}
            a \arrow[r, "e" description] \arrow[d, "e'" description]                & \bullet \arrow[d, "m" description] \\
            \bullet \arrow[r, "m'" description] \arrow[ru, "g" description, dotted] & b                                 
            \end{tikzcd}
        \end{center}
        \item Both $\cE$ and $\cM$ contain all isomorphisms and are closed under composition.
    \end{enumerate}

    A factorization system is called proper, if all morphisms in $\cE$ are epi and all morphisms in $\cM$ are mono.
\end{definition}

One example of such a factorization system in the category of finite undirected graphs is given by choosing $\cE$ as the (vertex) surjective graph homomorphisms, and $\cM$ as the strong graph embeddings.
The $\cM$-subobjects of $H$ in $G$ are precisely the induced subgraphs of $G$ that are isomorphic to $H$.
See \S\ref{sec:cat:ord_graphs} for more details on this.

Except the properties directly visible in Definition \ref{def:factorization_system}, we need one derived property (see \cite{riehl2008factorization}[Lemma 1.13] for a proof):
\begin{fact}[Cancellation Properties]
    Let ($\cE$, $\cM$) be a factorization system and $f: a \to b$ and $g: b \to c$ be morphisms.
    Then we have the following two implications:
    \begin{itemize}
        \item If $gf \in \cE$ and $f \in \cE$, then $g \in \cE$.
        \item If $gf \in \cM$ and $g \in \cM$, then $f \in \cM$.
    \end{itemize}
\end{fact}

Note that this implies that in a ($\cE$, $\cM$) factorization system, for $\cM$-subobjects $f$ and $f'$, whenever $f' \prec f$, i.e.\ if there is some morphism $g$ with $f' = f g$, then $g$ is already in $\cM$.
In other words:
Being a $\cM$-subobject is well-behaved under composition and cancellation.

Further, there are no non-trivial morphisms $f: a \to a$ in $\cM$:
\begin{lemma}
    \label{lem:cat_self_Ms_are_isos}
    Let $C$ be a locally small, finitely $\cM$-well-powered category and a proper ($\cE$, $\cM$)-factorization system.
    Then any morphism $f: a \to a$ in $\cM$ is already an isomorphism.
\end{lemma}
\begin{proof}
    We consider the set of morphisms $\{f^n \,:\, n \in \IN\}$.
    Since $\cM$ is closed under composition, all these morphisms are in $\cM$ and thus are $\cM$-subobjects of $a$.
    By the pidgeonhole principle, combined with $C$ being finitely $\cM$-well-powered, there must be some $n < m \in \IN$ with $f^n$ and $f^m$ being the same $\cM$-subobject, i.e.\ $f^n = f^m g$ with some isomorphism $g: a \to a$.
    This leads to the following diagram with $f^n\identity_a = f^m g$:
    \begin{center}
        \begin{tikzcd}
        a \arrow[r, "f^{m-n}g", shift left] \arrow[r, "\identity_a"', shift right] & a \arrow[r, "f^n"] & a
        \end{tikzcd}
    \end{center}
    Due to $f^n \in \cM$ and the factorization system being proper, we conclude that $f^n$ is a monomorphism and $f^{n-m} g = \identity_a$, which is equivalent to $f^{n-m} = g^{-1}$.
    It is now obvious that $f^{n-m-1}g = gf^{n-m-1}$ is the inverse of $f$.
\end{proof}

Formally, it is not clear how to attach an object from some category $C$ as an oracle to an NTM.
In order to do so, we denote by $\catsize{C}{j}$ some set of objects from $C$ that can be encoded as a $\polylog(j)$-bit oracle, i.e.\ as a subset of $\{0,1\}^{\bitsize(j)}$ for some efficiently computable function $\bitsize \in \polylog(j)$.
We require that for every object $c$ in $C$, there is some $j \in \IN$ and a $c' \in \catsize{C}{j}$, such that $c$ and $c'$ are isomorphic.

Given a nondeterministic Turing machine $M$ and an oracle object $c \in \catsize{C}{j}$,
we are interested in the number of accepting paths of $M$ when given oracle access to the $j$-bit oracle encoding $c$.
We define the number of accepting paths of $M$ on input $j \in \IN$ (given in binary\footnote{The precise nature of how $j$ is given, what the bitsize of the encodings of elements from $\catsize{C}{j}$ is and even whether the oracle is binary or not does not matter. The main importance is, that in polynomial time in the length of the input $M$ can query the oracle at polynomially many positions out of exponentially many possible positions.}) with oracle access to $c$ as $\#\acc_{M^c}(j)$.
Further we will also call $\#\acc_{M^c}(j)$ a computation.

Our setup has two categories $C$ and $P$, the ambient category and the category of \emph{pure} objects respectively.
While we list here most of the general properties of both categories we assume in general, each lemma still explicitly lists the properties it uses to make them easier to follow.
Both categories are locally small and $P$ is a full subcategory\footnote{Reminder, a subcategory is full if for all objects $a$ and $b$ that are in $P$, all morphisms between them in $C$ are also part of $P$.} of $C$, closed under isomorphisms, i.e.\ if $a$ and $b$ are isomorphic objects in $C$, they are either both included in $P$ or neither of them is.
Moreover there is a compatible proper ($\cE$, $\cM$)-factorization system of $C$ and $P$.
By compatible we mean some slight abuse of notation:
Not all morphisms in $\cM$ need to be part of $P$, however the restrictions of $\cE$ and $\cM$ to the morphisms in $P$ form themselves a proper factorization system of $P$.
To simplify notation we also call this restriction a proper ($\cE$, $\cM$)-factorization system of $P$.
Furthermore $C$ and $P$ are both finitely $\cM$-well-powered and both have the joint $\cM$-embedding property:
\begin{definition}
    \label{def:cat_joint_embedding}
    A locally small category $C$ with a ($\cE$, $\cM$)-factorization system has the joint $\cM$-embedding property if for every two objects $a$ and $b$, there is an object $c$ such that there are $\cM$-subobjects $f: a \to c$ and $g: b \to c$.
\end{definition}
Note that the resulting objects for the joint $\cM$-embedding property of $P$ also have to be contained in $P$.

The $\cM$-well-poweredness now allows us to define motif parameters $\varphi$ parallel to Definition~\ref{def:graph-motif-parameter} as linear combinations of $\cM$-subobject counting functions $\cntMsub{c}{\starpad}$ for fixed objects $c$ in $C$.
Note that whenever $c \cong c'$, then clearly $\cntMsub{c}{\starpad} = \cntMsub{c'}{\starpad}$, so the different patterns in the linear combination are again always assumed to be pairwise non-isomorphic.
We call a pattern $P$-pure if they are in $P$ and consequently we call $\varphi$ $P$-pure if all its patterns are $P$-pure.
As usual $\eval{\varphi}$ denotes the counting function lifting $\varphi$ to the corresponding promise type-2 function.

As usual, for some subclass $D$ of objects from $C$ that is closed under isomorphisms, we say that a motif parameter $\varphi$ is \defn{$D$-good} if all its coefficients are nonnegative integers and all its patterns are from $D$.
If any such coefficient is negative it is instead called \defn{$D$-bad}.
If the domain $D$ is clear from the context, we will simply call them \defn{good} and \defn{bad} respectively.

Before we get into the proof of our separation, we want to explain why we are restricting our motif parameters to be linear combinations.
If there are only finitely many morphisms in $\cE$ out of every object, then every monomial $\prod_{i=1}^k \cntMsub{a_i}{t}$ can be written as a linear combination with non-negative integer coefficients.
By extension, every polynomial over $\cM$-subobject counting functions can be written as a linear combination.
We prove the statement for monomials by giving a characterization as sets, the wanted linearization then follows by taking the cardinality of all sets involved.
\begin{lemma}
    \label{lem:cat:linearize}
    Let $C$ be a locally small category with finite coproducts and an ($\cE$, $\cM$)-factorization system.
    Then, assuming the axiom of choice, for any objects $a_1, \ldots, a_k, t$ from $C$ we have $\bigtimes_{i=1}^k \Msub{a_i}{t} \cong \bigsqcup_{b} \Msub{b}{t} \times \{ (h_1, \ldots, h_i) \in \bigtimes_{i=1}^k \Msub{a_i}{b} \mid \bigsqcup_{i=1}^k h_i \in \cE\}$ where $b$ ranges over all non-isomorphic objects in $C$.
\end{lemma}
\begin{proof}
    For each isomorphism class pick one designated object $b$ and for each $\cM$-subobject of $b$ in $t$, choose one designated morphism from $\cM$, both using the axiom of choice.
    
    We now directly give the bijection $\phi: \bigtimes_{i=1}^k \Msub{a_i}{t} \to \bigsqcup_{b} \Msub{b}{t} \times \{ (h_1, \ldots, h_i) \in \bigtimes_{i=1}^k \Msub{a_i}{b} \mid \bigsqcup_{i=1}^k h_i \in \cE\}$:\\
    Given $\cM$-subobjects $f_1, \ldots, f_k$ with $f_i \in \Msub{a_i}{t}$ for each $i$, construct $f = \bigsqcup_{i=1}^k f_i$.
    Using the ($\cE$, $\cM$)-factorization, factorize $f$ through some object $b$ via $e: \bigsqcup_{i=1}^k a_k \to b$ from $\cE$ and $m: b \to t$ from $\cM$.
    W.l.o.g.\footnote{Remember, $\cE$ and $\cM$ are both closed under composition and contain all isomorphisms. Further, the factorization is unique up to isomorphisms.}\ we can assume both $b$ and $m$ to be the corresponding designated object and morphism chosen by the axiom of choice.
    Since the factorization is unique up to isomorphism, this makes $(b, m)$ unique.
    This leads to the following commutative diagram for each $i$:
    \begin{center}
        \begin{tikzcd}
                                      & t                                                                                     &                                                         \\
        b \arrow[ru, "m" description] &                                                                                       & a_i \arrow[lu, "f_i" description] \arrow[ld, "\iota_i"] \\
                                      & \bigsqcup_{i=1}^k a_i \arrow[lu, "e" description] \arrow[uu, "f = \bigsqcup_{i=1}^k f_i" description] &                                                        
        \end{tikzcd}
    \end{center}
    Since $f_i$ and $m$ are both in $\cM$, this implies that $h_i := e \iota_i$ is in $\cM$ by the cancellation property of a factorization system.
    Due to $e = \bigsqcup_{i=1}^k h_i$, we see that $\bigsqcup_{i=1}^k h_i$ is in $\cE$.
    Our bijection thus maps $(f_1, \ldots, f_k)$ to $(m, (h_1, \ldots, h_k))$.
    Note that this is well-defined: if there are morphisms $f'_1, \ldots, f'_k$ from $\cM$ and isomorphism $g_1, \ldots, g_k$ with $f'_i = f_i g_i$, then we get $h'_i = h_i g_i$ and thus output the same $\cM$-subobjects.

    To define $\Phi^{-1}$, let $b$ be some designated object in $C$, let $m$ be the designated morphism of some $\cM$-subobject of $b$ in $t$ and let $(h_1, \ldots, h_k) \in \bigtimes_{i=1}^k \Msub{a_i}{b}$ with $\bigsqcup_{i=1}^k h_i \in \cE$.
    Since $\cM$ is closed under composition, we see that each of the $mh_i$ is in $\cM$.
    The bijection then outputs $(mh_1, \ldots, mh_k)$.
    This function is well-defined by the same argument as above.

    Bijectivity now follows from the fact that $\bigsqcup_{i=1}^k f_i = m \bigsqcup_{i=1}^k h_i$.
\end{proof}

Even more, we can even prove that the coefficients in the linear combination are non-negative integers if the polynomial, given in the binomial basis, has non-negative integer coefficients (potentially leading to negative coefficients in front of the polynomial in monomial basis).
For example $\varphi(t) = \binom{\cntMsub{a}{t}}{2} = \frac{(\cntMsub{a}{t})^2}{2} - \frac{\cntMsub{a}{t}}{2}$ can be written as a linear combination of $\cM$-subobject counting functions with non-negative integer coefficients.

\begin{corollary} \label{cor:cat:lin}
    Let $C$ be a locally small category with finite coproducts and a proper ($\cE$, $\cM$)-factorization system.
    Then, assuming the axiom of choice, for any objects $a$ and $t$ from $C$ we have $\binom{\Msub{a}{t}}{k} \cong \bigsqcup_{b} \Msub{b}{t} \times \{ (h_1, \ldots, h_i) \in \binom{\Msub{a}{b}}{k} \mid \bigsqcup_{i=1}^k h_i \in \cE\}$ where $b$ ranges over all non-isomorphic objects in $C$.
\end{corollary}
\begin{proof}
    In the proof of Lemma~\ref{lem:cat:linearize} in the construction of the bijection, whenever $\cM$ only contains monos, from $f_i = f_j$, we automatically obtain $h_i = h_j$ and vice-versa.
    Restricting to only the cases where the morphisms are pairwise disjoint gives the wanted bijection.
\end{proof}

We now continue on to the proof of our separation by first generalizing the notion of set-instantiators:
\begin{definition}[Set-instantiator against $(M,c,\varphi)$]
\label{def:cat_setinstantiator}
Let $C$ and $P$ be locally small, finitely $\cM$-well-powered categories with a compatible ($\cE$, $\cM$)-factorization system and where $P$ is a full subcategory of $C$ that is closed under isomorphism.

Let $M$ be a nondeterministic Turing machine,
let $c \in C$ and
let $\varphi$ be a $P$-pure motif parameter.

Let  $\top$  be a symbolic top element above the $\cM$-subobject poset $\poset{c}$, i.e., $\top \not\prec f$ for all $\cM$-subobjects $f \in \poset{c}$.
A  \defn{set-instantiator}  $\SI$  is a triple of

$\circ$ \ some $j \in \IN$, 

$\circ$ \ an instantiation function  $\inst_\SI : \poset{c} \to \catsize{C}{j}$,   and

$\circ$ \  a perception function  $\perc_\SI : \{0,1\}^* \to \poset{c} \cup \{\top\}$,

\smallskip
\noindent
such that the following property holds for all $\cM$-subobjects $f \in \poset{c}$.

\smallskip
$\bu$ \ $\tau\in\{0,1\}^*$ is an accepting path for the computation  $\#\acc_{M^{\inst_\SI(f)}}(j)$
iff $\perc_\SI(\tau) \prec f$, and 

$\bu$ \ $\varphi(\inst_\SI(f))) = \varphi(\domain f)$.
\end{definition}

Again, we think of perceptions of an accepting computation path of $M$ to be the $\cM$-subobject of $c$ that is observed by $M$ through the oracle queries, while computation paths that do not accept are given perception $\top$.
Unfortunately there is no generalized construction for set-instantiators as this has to heavily depend on the actual encoding of the objects as oracle layers.
For example consider an encoding where every object is encoded as an oracle layer containing exactly one $1$.
Then $M$ can determine all information about the object by guessing the position of that single $1$ and then querying the oracle once.
However, this encoding is not particularly natural and set-instantiators often do exist when encoded naturally.
See sections \S\ref{sec:cat:finvec} and \S\ref{sec:cat:parsets} for two more set-instantiators.

Again, we have the problem that for $\cM$-subobjects $f$ and $g$ with $\domain f \cong \domain g$ we could have $\#\acc_{M^{\inst_\SI(f)}}(j) \neq \#\acc_{M^{\inst_\SI(g)}}(j)$, which we do not want to happen in a good function.
We thus want to construct an instantiation function that no longer has this problem.
We achieve this by creating a very large set-instantiator and then restricting ourselves to some small part of the set-instantiator where we can enforce this property.

Ensuring this structure is achieved by requiring $C$ to have the Ramsey property, relative to $\cM$-subobjects.
\begin{definition}
    \label{def:cat_ramsey_prop}
    A locally small category $C$ with a ($\cE$, $\cM$)-factorization system is called $\cM$-Ramsey, if for every $a, b \in C$ and $t \in \IN$, there is some $c \in C$, such that for every $\Phi: \Msub{a}{c} \to \{0, \ldots, t\}$, there is a $f_\Phi \in \Msub{b}{c}$ with the property that $\Phi$ is constant when restricted to inputs from $\{g \in \Msub{a}{c} \,: \, g \prec f_\Phi\}$.
\end{definition}
See~\cite{graham1972ramsey} for some categorical conditions that imply the Ramsey property.

We extend this to colorings of all $\cM$-subobjects of $b$, not just those with their domain being $a$:
\begin{proposition}[Ramsey Theorem]
    \label{pro:cat_ramseyposet}
    Let $C$ be an $\cM$-Ramsey, locally small and finitely $\cM$-well-powered category with a ($\cE$, $\cM$)-factorization system.
    
    Let $b \in C$ and let $t \in \IN$ be fixed.
    Then there is a $c \in C$, such that for every $\Phi: \poset{c} \to \{0, \ldots, t\}$ there is a $f_\Phi \in \Msub{b}{c}$ with the property, if $g, h \prec f_\Phi$ with $\domain g \cong \domain h$, then $\Phi(g) = \Phi(h)$.
\end{proposition}
\begin{proof}
    The $\poset{b}$ is finite due to $C$ being finitely $\cM$-well-powered.
    Now inductively apply the property that $C$ is $\cM$-Ramsey on $\domain \poset{b}$.
\end{proof}

This theorem can now be used to ensure that for every polynomial-time NTM there must be some small\footnote{Small in the colloquial sense of small, meaning some ``small'' finite set, not in the sense of category theory where small simply means a set.} set of inputs where the NTM behaves like a good function.
The instantiation function for the inputs where the NTM behaves like a good function now doesn't instantiate individual subobjects anymore, but instead only instantiates objects (the domains of the subobjects).
This change is natural since the computed function is now locally a motif parameter and thus invariant under isomorphisms, i.e.\ the precise subobject doesn't matter, only its domain).

\begin{lemma}
    \label{lem:cat_forceGood}
    Let $C$ and $P$ be locally small and finitely $\cM$-well-powered categories with a compatible ($\cE$, $\cM$)-factorization system, where $P$ is a full subcategory of $C$ that is closed under isomorphisms and where $C$ is $\cM$-Ramsey.
    
    Let $\varphi$ be a $P$-pure motif parameter, let $M$ be any nondeterministic polynomial-time Turing machine computing $\eval{\varphi}$ and let $b \in C$.
    If for every $c \in C$ there are set-instantiators against $(M, c, \varphi)$,
    then there is some $j \in \IN$, a function $\inst: \domain \poset{b} \to \catsize{C}{j}$ and a $\domain \poset{b}$-good function $\Psi$ with $\#\acc_{M^{\inst(a)}}(j) = \Psi(a)$ and $\varphi(\inst(a)) = \varphi(a)$ for every $a \in \domain \poset{b}$.
\end{lemma}
\begin{proof}
    We invoke Proposition~\ref{pro:cat_ramseyposet} for $b$ and $t := \max\{\varphi(a) \,:\, a \in \domain \poset{b}\}$ to obtain $c \in C$.
    We now want to color $\poset{c}$ according to the behaviour of $M$ on the elements of it, in particular how many accepting paths query each specific $\cM$-subobject of $c$.
    Before we can do this we have to first embed all potential $\cM$-subobjects in $\poset{c}$ as oracles.
    In order to do this we construct a set-instantiator $\SI$ against $(M, c, \varphi)$ to obtain a $j \in \IN$, $\inst_\SI$ and $\perc_\SI$ as in Definition~\ref{def:cat_setinstantiator}.
    
    For $g \in \poset{c}$ define $\Phi(g) := \#\acc_{M^{\inst_\SI(g)}}(j)$.
    Note that
    \begin{equation}\label{eq:cat_PhiPerception}
        \Phi(g) \ = \ \bigl|\bigl\{\tau\in\{0,1\}^* \,: \, \perc_\SI(\tau) \prec g\bigr\}\bigr| = \sum_{h \prec g} \, \bigl|\big\{\tau\in\{0,1\}^* \,:\, \textsu{perc}_{\SI}(\tau)=h\big\}\bigr|.
    \end{equation}

    Using $\Phi$ as a coloring for Proposition~\ref{pro:cat_ramseyposet}, we obtain $f_\Phi$ with $\domain f_\Phi = b$.
    Note that $\Phi(g)$ now only depends on the isomorphism type of the domain of $g$ for all $g \prec f_\Phi$.
    By induction we see that the number
    \[
        \big|\big\{\tau\in\{0,1\}^* \, :\, \textsu{perc}_{\SI}(\tau)=g\big\}\big|
    \]
    also depends only on the isomorphism type of $\domain g$, whenever $g \prec f_\Phi$.
    Denote this number by $\alpha_{\domain g}$ and set $\Psi(\domain g) := \Phi(g)$.
    The function $\Phi$ is well-defined and invariant under isomorphisms of $\domain g$ only for $g \prec f_\Phi$, thus we consider the domain of $\Phi$ to be $\{\domain g \,:\, g \prec f_\Phi\} = \domain \poset{b}$.
    This means that \eqref{eq:cat_PhiPerception} simplifies:
    \begin{equation}\label{eq:cat_PhiGood}
    \Psi(\domain g) \ = \ \sum_{\domain h} \cntMsub{\domain h}{\domain g} \alpha_{\domain g}
    \end{equation}
    where the sum is over all non-isomorphic domains $\domain h$ for $h \in \poset{b}$.
    This implies that $\Phi$ is $\domain \poset{b}$-good.

    Remains to define the function $\inst: \domain \poset{b} \to \catsize{C}{j}$.
    Let $a \in \domain \poset{b}$, then $\inst(a) := \inst_{\SI}(g_\Phi)$ for some $g_\Phi \prec f_\Phi$ with $\domain g_\Phi \cong a$.
    The property $\varphi(\inst(a)) = \varphi(a)$ now directly follows from $\SI$ being a set-instantiator.
\end{proof}

Note that Lemma~\ref{lem:cat_forceGood} only guarantees the local behaviour of $M$ to be $\domain \poset{b}$-good.
This leaves two possibilities for contradictions: Either the local behaviour is even $(\domain \poset{b} \cap P)$-good, or the function grows too quickly for non $P$-pure inputs.
This next Witness Theorem shows that in the first case we can always find an object that proves that $M$ does not compute the correct function using a simple linear algebra argument.

\begin{theorem}[The Witness Theorem]
    \label{thm:cat_findcounterexample}
    Let $P$ be a locally small, finitely $\cM$-well-powered category with a proper ($\cE$, $\cM$)-factorization system and the joint $\cM$-embedding property.
    
    Fix a bad $\varphi: P \to \IN$.
    Then there exists $b \in P$ such that for
    every $(\domain\poset{b} \cap P)$-good function $\Psi:(\domain\poset{b} \cap P) \to \IN$,
    there exists $w \in (\domain\poset{b} \cap P)$
    with $\Psi(w) \neq \varphi(w)$.
\end{theorem}
\begin{proof}
    Use the joint embedding property of $P$ on $\supp(\varphi)$ to obtain an object $b \in P$, such that there are morphisms $a \to b$ in $\cM$ for all $a \in \supp(\varphi)$.
    We now consider $P' := \domain\poset{b} \cap P$, keeping only one representative per isomorphism class.
    We now prove that the functions $\cntMsub{a}{\starpad}$ for $a \in P'$ are linearly independent:
    Consider the evaluation matrix $A$ with rows and columns indexed by objects from $P'$, sorted by a total extension of the poset $\poset{b}$  and values $A_{ij} := \cntMsub{i}{j}$.
    This is well-defined by Lemma~\ref{lem:cat_self_Ms_are_isos}: if there are morphisms from $\cM$, $f: a \to a'$ and $g: a' \to a$, then $a$ and $a'$ are isomorphic and thus only one of them is part of $P'$.
    If this matrix is invertible this directly implies that the $\cntMsub{a}{\starpad}$ (i.e.\ the rows) are linearly independent.
    
    We now prove that $A$ is an upper-diagonal matrix with a $1$s on the diagonal and thus invertible.
    The matrix $A$ is upper-diagonal since we've ordered the rows and columns accordingly, combined with Lemma~\ref{lem:cat_self_Ms_are_isos}.
    Remaining are the $1$s on the diagonal:
    For any object $a$, the identity morphisms $\identity_a$ is contained in $\cM$ and thus there is at least one $\cM$-subobject in $\Msub{a}{a}$.
    Furthermore every morphism $f: a \to a$ in $\cM$ is an isomorphism by Lemma~\ref{lem:cat_self_Ms_are_isos} and thus the same $\cM$-subobject as $\identity_a$.

    We finish the proof of the Witness Theorem by noticing that the linear independence implies that the patterns and coefficients of any function $\Psi$ with support $\supp(\Psi) \subseteq P'$ are uniquely determined by the evaluations $\Psi$ on $P'$.
    Thus there exists a witness $a \in P'$ with $\Psi(a) \neq \varphi(a)$, since $\Psi$ is good while $\varphi$ is bad (thus, in particular, $\Psi$ and $\varphi$ are not the same function).
\end{proof}

We need one last condition on $C$ and $P$:
A procedure that, given some $c$ and some $k \in \IN$ allows to construct a ``bigger'' object $c'$ that cannot be distinguished by any $P$-pure motif parameters, but every non $P$-pure basis function with patterns from $\domain\poset{c}$ evaluates to a value bigger than $k$.
Formally:
\begin{definition}
    Let $C$ and $P$ be locally small and finitely $\cM$-well-powered categories with a compatible ($\cE$, $\cM$)-factorization system, where $P$ is a full subcategory of $C$ that is closed under isomorphisms.

    We say that $C$ and $P$ fulfill the \defn{blowup property}, if for every object $c$ in $C$ and $k \in \IN$ there is some $c'$ in $C$, such that
    \begin{align*}
        \cntMsub{a}{c'} &> k && \text{for every $a \in \domain\poset{c} \setminus P$}\\
        \cntMsub{a}{c'} &= \cntMsub{a}{c'} && \text{for every $a \in P$}
    \end{align*}
\end{definition}
In the case of graphs, this was achieved by adding isolated vertices.

With this last tool we are finally able to lead the existence of polynomial-time NTMs computing bad ordered graph motif parameters to a contradiction.

\begin{lemma} \label{lem:cat:prop}
    \label{lem:cat_seperate}
    Let $C$ and $P$ be two categories with the following properties:
    \begin{itemize}
        \item $P$ is a full subcategory of $C$, closed under isomorphism.
        \item $C$ and $P$ have a compatible proper ($\cE$, $\cM$)-factorization system.
        \item $C$ and $P$ are both locally small.
        \item $C$ and $P$ are both finitely $\cM$-well-powered.
        \item $C$ and $P$ both have the joint $\cM$-embedding property.
        \item $C$ is $\cM$-Ramsey.
        \item $C$ and $P$ fulfill the blowup property.
    \end{itemize}

    Let $M$ be any nondeterministic polynomial-time Turing machine computing $\eval{\varphi}$ for some bad $P$-pure motif parameter $\varphi$.
    If for every $c \in C$ there are set-instantiators against $(M, c, \varphi)$, then there is a $j \in \IN$ and $c \in \catsize{C}{j}$ such that $\#\acc_{M^c}(j) \neq \varphi(c)$.
\end{lemma}
\begin{proof}
    We invoke the Witness Theorem~\ref{thm:cat_findcounterexample} with our $\varphi$ 
    to obtain a $c_1 \in P$ as in the theorem.
    Using the blowup property, construct $c_2$ such that
    \begin{align}
        \cntMsub{a}{c_2} &> \max\big\{\varphi(b) \,:\, b \in \domain\poset{c_1} \cap P \big\} && \text{for every $a \in \domain\poset{c_1} \setminus P$} \label{eq:cat_blowupNonR}\\
        \cntMsub{a}{c_2} &= \cntMsub{a}{c_1} && \text{for every $a \in P$}\,. \label{eq:cat_keepR}
    \end{align}

    We now use Lemma~\ref{lem:cat_forceGood} to obtain a $j \in \IN$, a function $\inst: \domain\poset{c_2} \to \catsize{C}{j}$ and a $\domain\poset{c_2}$-good function $\Psi$ with $\#\acc_{M^{\inst(a)}}(j) = \Psi(a)$ and $\varphi(\inst(a)) = \varphi(a)$ for every $a \in \domain\poset{c_2}$.

    There are now two possibilities to find the required counterexample $w$:
    If there are no basis functions present in $\Psi$ with a pattern in $\domain\poset{c_1} \setminus P$, we can invoke the Witness Theorem again, however if any of these basis functions are present we can construct a direct counterexample.

    We first handle the case where such a basis function is present, so let the coefficient of $\cntMsub{a}{\starpad}$ in $\Psi$ be positive for some $a \in \domain \poset{c_1} \setminus P$.
    Then
    \[
        \Psi(c_2) \geq \cntMsub{a}{c_2} > \max\big\{\varphi(b) \,:\, b \in \domain\poset{c_1} \cap P \big\} \geq \varphi(c_1) = \varphi(c_2)
    \]
    by \eqref{eq:cat_blowupNonR} and \eqref{eq:cat_keepR}, combined with the fact that $\varphi$ is $P$-pure and thus only has patterns from $P$.
    Our counterexample in this case is $w := c_2$.
    
    Otherwise assume that no such basis function is present.
    Then $\tilde\Psi$ obtained from $\Psi$ by restricting to $\domain\poset{c_1}$ (i.e.\ setting all coefficients of $\domain \poset{c_2} \setminus \domain \poset{c_1}$ to zero), is $(\domain\poset{c_1} \cap P)$-good, which is a requirement for the second part of the Witness Theorem~\ref{thm:cat_findcounterexample} (we already used it to obtain $c_1$).
    Furthermore for the restriction to $\domain\poset{c_1}$ we have $\tilde \Psi_{|\domain\poset{c_1}} = \Psi_{|\domain\poset{c_1}}$.
    We invoke the second part of the Witness Theorem~\ref{thm:cat_findcounterexample} to find a point 
    $w \in \domain \poset{c_1} \cap P$ with $\Psi(w) = \tilde\Psi (w) \neq \varphi(w)$ which is our counterexample.
    
    Independent of how we have obtained our $w$ we observe
    \begin{equation}
        \label{eq:cat_counterexamplepoint}
        \#\acc_{M^{\inst(w)}}(j) = \Psi(w) \neq \varphi(w),
    \end{equation}
    which completes the proof.
\end{proof}

Diagonalizing over all polynomial-time NTMs computing $\eval{\varphi}$ now leads to our general hardness result:
\begin{theorem}
    \label{thm:cat_main_theorem}
    Let $C$ and $P$ be as in Lemma~\ref{lem:cat_seperate} and let $\varphi$ be a $P$-pure motif parameter.
    If set-instatiators against $(M, c, \varphi)$ exist for all nondeterministic polynomial-time Turing machines $M$ computing $\eval{\varphi}$ and objects $c \in C$, then $\eval{\varphi} \notin \promise\sharpPdottedcircle$ unless $\varphi$ is good.
\end{theorem}
\begin{proof}
    There are two possibilities for $\varphi$ to not be good:
    Either $\varphi$ is bad or $\varphi$ has a non-integer coefficient.

    We first treat the case where $\varphi$ is bad.
    Assume for the sake of contradiction $\eval{\varphi} \in \promise\sharpPdottedcircle$.
    Then there is a polynomial-time NTM $M$ computing $\eval{\varphi}$.
    Due to the existence of set-instantiators we can invoke Lemma~\ref{lem:cat_seperate} and obtain $j \in \IN$ and $w \in \catsize{C}{j}$ with $\#\acc_{M^w}(j) \neq \varphi(w)$.
    This however, means that $M$ in fact does not compute $\eval{\varphi}$, a contradiction.

    Second we treat the case where $\varphi$ has a non-integer coefficient.
    Similarly to the proof of Theorem~\ref{thm:cat_findcounterexample}, using the joint $\cM$-embedding property of $C$ on $\supp(\varphi)$ to obtain an object $b \in C$, such that there are morphisms $a \to b$ in $\cM$ for all $a \in \supp(\varphi)$.
    Using the $\prec$ order, find some minimal element $f$ in $\poset{b}$, such that the pattern $\domain f$ has a non-integer coefficient $\alpha$.
    Consider $\varphi(\domain f)$.
    Recall the evaluation matrix $A$ from the proof of Theorem~\ref{thm:cat_findcounterexample}.
    Since it is upper-diagonal with $1$s on the diagonal, we see that the only patterns affecting $\varphi(\domain f)$ are the ones with patterns $\domain g$, for $g \prec f$.
    All of them have integer coefficients and thus $\varphi(\domain f) = n + \alpha \cdot \cntMsub{\domain f}{\domain f}$ for some $n \in \IZ$.
    The $1$s on the diagonal correspond to $\cntMsub{\domain f}{\domain f} = 1$, so $\varphi(\domain f) = n + \alpha \notin \IZ$.
    We conclude, $\varphi$ does not only output integers and thus $\eval{\varphi} \notin \promise\sharpPdottedcircle$.
\end{proof}

In our general theorem above, we do not get a dichotomy, since we do not need to have an upper bound. It might be the case that the evaluation of the motif parameter is not contained in $\promise\sharpPdottedcircle$ even when the coefficients are all non-negative integers.
However, in our applications below it is easy to show that this does not happen.

\subsection{Ordered Graphs}
\label{sec:cat:ord_graphs}
To show how this generalization compares to the proof of Theorem~\ref{thm:ordgraphs}, we first show how this categorical generalization can be used to reprove the separation of Theorem~\ref{thm:ordgraphs}:

We consider $C = \finordgraphs$ to be the category of all finite ordered graphs.
The objects of this category are the finite\footnote{The finite was implicit throughout this entire paper, but we make it explicit here to distinguish of the category of ordered graphs, which usually includes infinite graphs.} ordered graphs while the morphisms are the monotone graph homomorphisms.
The monomorphisms of $\finordgraphs$ are the injective monotone graph homomorphisms, i.e.\ (not necessarily strong) monotone graph embeddings.
The epimorphism are the vertex surjective monotone graph homomorphisms, i.e.\ $f: H \to G$ is epi, if for every $v \in V(G)$ there is some $v' \in V(H)$ such that $f(v) = v'$.
The category $P$ of pure objects is obtained by restricting $C$ to those ordered graphs without isolated vertices and keeping all morphisms between such ordered graphs, making $P$ a full subcategory of $C$.
Clearly $P$ is closed under isomorphisms, the property of having isolated vertices is preserved by isomorphisms.
Both categories are locally finite, there are only finitely many monotone graph homomorphisms between two fixed graphs, and thus locally small.

Next we describe the proper ($\cE$, $\cM$)-factorization system:
We choose $\cE$ as the epimorphisms and we choose $\cM$ as the strong monotone graph embeddings.
This makes the $\cM$-subobjects of some monotone graph $G$ precisely the induced subgraphs of $G$ (or rather an equivalence class, represented by each of the induced subgraphs).
We prove this indeed forms a factorization system:

For that let $f: H \to G$ be an monotone graph homomorphism.
Consider $G[f(V(H))]$, the subgraph of $G$ induced by the image of the vertices of $H$ under $f$.
Let $m: G[f(V(H))] \to G$ be the canonical strong monotone graph embedding (the identity) and let $e: H \to G[f(V(H))]$ be obtained from $f$ by restricting its codomain.
Then $e$ is surjective on the vertices and thus $e \in \cE$ and $m \in \cM$ with $f = me$.
On the other hand, let $f = m'e'$ be a different factorization with $e': H \to F$ in $\cE$ and $m': F \to G$ in $\cM$ through some ordered graph $F$.
Due to $m' \in \cM$, $F$ is isomorphic to some induced subgraph of $G$, say $F \cong G[X]$ with isomorphism $g: F \to G[X]$ for some $X \subseteq V(G)$.
W.l.o.g.\ assume the strong monotone graph embedding $G[X] \to G$ be the identity (extended to all of $G$ as a codomain).
Then $X = ge'(V(H)) = f(V(H))$ and thus $mge' = me$.

Compatibility of the factorization system between $C$ and $P$ is given by the following fact:
Let $f: H \to G$ factor as $f = me$ with $e: H \to F$ in $\cE$ and $m: F \to G$ in $\cM$ and let $v \in V(F)$ be an isolated vertex in $F$.
Then there is an $v' \in V(H)$ with $e(v') = v$ due to vertex surjectivity of $e$.
The images of all neighbors of $v'$ in $H$ under $e$ are neighbors of $v$ in $F$, due to $e$ being a homomorphism.
Thus $v'$ has to be an isolated vertex in $H$.
Together this implies that $F$ is pure if $H$ is pure.

There are only finitely many induced subgraphs for any ordered graph $G$, making $C$ and $P$ both $\cM$-well-powered.
The joint $\cM$-embedding property for ordered graphs $G, H$ is given by the disjoint union of $G$ and $H$ with the obvious inlusions.
The category $C$ being $\cM$-Ramsey is given by Lemma~\ref{lem:nesetril}.
The blowup property is fulfilled by adding isolated vertices to the ordered graph we are trying to blow up.

The sets $\catsize{C}{j}$ are precisely the ordered graphs $G$ with $V(G) = \{1, \ldots, j\}$, represented as oracles as before.
The corresponding set-instantiators have been constructed in Lemma~\ref{lem:graphs_setinstantiators_exist}.

Theorem~\ref{thm:cat_main_theorem} then reproves the separation of Theorem~\ref{thm:ordgraphs}.

\subsection{Finite Vector Spaces over Finite Fields}
\label{sec:cat:finvec}
Our first new example using the categorical generalization is the category $C := \finvec_{\IF_p}$ of finite vector spaces over a fixed finite field $\IF_p$.
Here, an NTM is given the size $p^d$ of an ambient vector space $\IF_p^d$ and some vector space $V \subseteq \IF_p^d$.
It can query $V$ by querying the oracle for individual vectors, i.e.\ it can ask whether $(v_1, \ldots, v_d) \in V$ for $v_1, \ldots, v_d \in \IF_p$.
The vector $(v_1, \ldots, v_d)$ is suitably encoded as a binary string of some fixed length $\bitsize(d)$, only depending on $p$ and $d$, and thus $V$ itself can be encoded as a subset of $\{0, 1\}^{\bitsize(d)}$.
Our sets $\catsize{C}{p^d}$ are then precisely all the vector spaces $V \subseteq \IF_p^d$ and $\bitsize(d)$ is poly-logarithmic in $p^d$.

Due to the ambient space we can already pad a vector space by embedding it into a larger ambient space.
This allows us to choose $P$ and $C$ to be the same category $\finvec_{\IF_p}$ that we now describe:
Let $\IF_p$ be some fixed finite field.
The objects of $\finvec_{\IF_p}$ are all the finite vector spaces over $\IF_p$ and the morphisms between them are the vector space homomorphisms.
This makes the epimorphisms the surjective homomorphisms, while the monomorphisms are the injective homomorphisms.
For some vector spaces $V, W \in \finvec_{\IF_p}$, the homomorphisms from $V$ to $W$ can be represented by the $\dim W \times \dim V$ matrices over $\IF_p$.
We thus conclude that $\finvec_{\IF_p}$ is locally small and even locally finite.
The proper ($\cE$, $\cM$) factorization system is given by choosing $\cE$ to be all epimorphisms and $\cM$ to be all monomorphisms.
That this indeed forms a factorization system follows from $\finvec_{\IF_p}$ being an abelian category, we refer the interested reader to \cite{mac1998categories}[Chapter VIII, Section 3, Proposition 1] for the details.
The $\cM$-subobjects of some vector space $V$ are represented precisely by the subspaces $W \subseteq V$, making $\finvec_{\IF_p}$ finitely $\cM$-well-powered.
We thus have $\cntMsub{W}{V} = \binom{\dim(V)}{\dim(W)}_p$, the Gaussian binomial coefficient (see \cite{Kac2002}[Theorem 7.1] for details on this).
As such the motif parameters $\varphi: \finvec_{\IF_p} \to \IQ$ simply depend on the dimension of the input.
Furthermore $P = C$ implies that the $P$-pure motif parameters are precisely the motif parameters.
We will thus only deal with motif parameters in the following.
The category $\finvec_{\IF_p}$ is well known to admit finite coproducts: the direct sum of vector spaces.
The coproduct injections for coproducts in $\finvec_{\IF_p}$ are always mono, implying the joint $\cM$-embedding property.
The Graham-Leeb-Rothschild Theorem \cite{graham1972ramsey} states that $\finvec_{\IF_p}$ is $\cM$-Ramsey.
The blowup property for $C$ and $P$ is vacuously true for $P = C$.

The only thing left to apply Theorem~\ref{thm:cat_main_theorem} is to prove the existence of set-instantiators:
\begin{lemma}
    Let $\varphi: \finvec_{\IF_p} \to \IN$ be a motif parameter, let $M$ be a polynomial-time NTM computing $\eval{\varphi}$ and let $V \in \finvec_{\IF_p}$.
    Then there is a set-instantiator against $(M,V,\varphi)$.
\end{lemma}
\begin{proof}
    For now let the dimension $d$ of the ambient space be undetermined, we will choose it accordingly later.
    Consider uniformly chosen vector space homomorphisms $h: V \to \IF_p^d$.
    Our goal is to use the union bound to show that there is a choice for $h$ such that:
    \begin{enumerate}
        \item $h$ is injective. \label{itm:vec_setinst_prop1}
        \item For each subspace $W \subseteq V$, all accepting paths of the computation $\#\acc_{M^{h(W)}}(p^d)$ do not query the oracle for any vectors in $h(V \setminus W)$. \label{itm:vec_setinst_prop2}
    \end{enumerate}
    If each of these events (we count property~\ref{itm:vec_setinst_prop2} individually for each of the finitely many subspaces) happens with high probability with growing $d$, then eventually there is an $h$ fulfilling all of them.

    For now, fix some subspace $W \subseteq V$.
    Further group the homomorphisms, by their image of $W$.
    Call one such group $H$.
    Then the set $S_H$ of all vectors queried by all accepting paths of the computation $\#\acc_{M^{h(W)}}(p^d)$ is the same for all $h \in H$.
    Choose a basis $w_1, \ldots, w_n$ of $W$ and extend it with vectors $v_1, \ldots, v_m$ to a basis of $V$.
    The images $h(v_1), \ldots, h(v_m)$ are now independently and uniformly distributed in $\IF_p^d$.
    Any fixed vector $v \in V \setminus W$ can be written as $\sum_{i=1}^{n}\alpha_i w_i + \sum_{i=1}^m\beta_i v_i$ with some $\beta_i \neq 0$ and thus $h(v)$ is also uniformly distributed in $\IF_p^d$.
    We thus have
    \begin{align*}
        \Pr_{h \in H}[h(v) \in S_H] &= \frac{|S_H|}{p^d}\,.
    \end{align*}
    Combining this using the union bound and $|V \setminus W| \leq |V|$ we get
    \begin{align*}
        \Pr_{h \in H}[\exists v \in V \setminus W \text{ with } h(v) \in S_H] &\leq |V| \cdot \frac{|S_H|}{p^d}\,.
    \end{align*}
    Remains to estimate $|S_H|$.
    Let $t_M(p^d)$ be the worst-case running-time of $M$ on input $p^d$.
    Each oracle query can query exactly one vector, there are at most $\varphi(h(W))$ many accepting paths and each accepting path can query the oracle at most $t_M(p^d)$ many times.
    This gives the bound $|S_H| \leq \varphi(h(W)) \cdot t_M(p^d)$ on the size of $S_H$.
    This bound is polynomial in $d$: $M$ is polynomial-time and $\varphi(h(W))$ is bounded by the constant $\max_{W' \subseteq V} \varphi(W')$ due to $\varphi$ being a motif parameter.
    
    Across all homomorphisms $h: V \to \IF_p^d$ (not just those from $H$) we thus get that all accepting paths of the computation $\#\acc_{M^{h(W)}}(p^d)$ do not query the oracle for any vectors in $h(V \setminus W)$ with a probability of at least $1 - |V| \cdot \frac{\poly(d)}{p^d}$.

    Remains to show that $h$ is also injective with high probability:
    Conversely, if $h$ is not injective, then there is some non-zero $v \in V$ with $h(v) = 0$.
    For a fixed $v \in V$, $h(v)$ is again uniformly distributed, so the chance of this happening for this specific $v$ is $\frac{1}{p^d}$.
    Using the union bound yet again, we get that the chance that $h$ is injective is at least $1 - |V| \cdot \frac{1}{p^d}$.

    Since $|V|$ is constant a final union bound proves that there is a large enough $d$, such that there is a homomorphism $h$ fulfilling properties~\ref{itm:vec_setinst_prop1} and~\ref{itm:vec_setinst_prop2}.

    We now construct our set-instantiator with $j=p^d$ for the $d$ chosen above.
    The instantiation function for a subspace $W \subseteq V$ is $\inst_{\SI}(W) = h(W)$.
    For any computation path $\tau \in \{0, 1\}^\star$ of a computation $\#\acc_{M^{h(W)}}(j)$ denote by $S_\tau \subseteq \IF_p^j$ the vectors that are queried by the computation.
    We set
    \begin{align*}
        \perc_{\SI}(\tau) := \begin{cases}
            \Span(h^{-1}(S_\tau)) & \text{if $\tau$ is an accepting path of the computation $\#\acc_{M^{h(V)}}(j)$,}\\
            \top & \text{otherwise}
        \end{cases}
    \end{align*}
    We check the properties of a set-instantiator.
    Clearly $\varphi(\inst_{\SI}(W)) = \varphi(W)$ for all subspaces $W \subseteq V$ due to $h$ being injective, making $h(W)$ and $W$ isomorphic.
    Further we need that for all subspaces $W \subseteq V$ we have that $\tau \in \{0, 1\}^\star$ is an accepting path for the computation $\#\acc_{M^{h(W)}}(j)$ iff $\perc_\SI(\tau) \subseteq W$.
    For this, fix some subspace $W \subseteq V$ and let $\tau$ be an accepting computation of the computation $\#\acc_{M^{h(W)}}(j)$.
    By our choice of $h$ we know that no other vertices of $h(V \setminus W)$ are queried, or in other words, $h^{-1}(S_\tau) \subseteq W$ and thus $\tau$ is also an accepting path for the computation $\#\acc_{M^{h(V)}}(j)$.
    This directly implies $\perc_{\SI}(\tau) = \Span(h^{-1}(S_\tau)) \subseteq W$.
    On the other hand, let $\perc_{\SI}{\tau} \subseteq W$, then $\perc_{\SI} \neq \top$ and $\tau$ is an accepting path of the computation $\#\acc_{M^{h(V)}}(j)$.
    We conclude that $\tau$ is also an accepting path of the computation $\#\acc_{M^{h(W)}}(j)$ due to $\perc_{\SI}(\tau) = \Span(h^{-1}(S_\tau)) \subseteq W$.
\end{proof}

Applying Theorem~\ref{thm:cat_main_theorem} now proves the main result for $\finvec_{\IF_p}$.
\begin{theorem} \label{thm:cat:main:vectorspace}
    Let $\varphi: \finvec_{\IF_p} \to \IN$ be a motif parameter.
    Then $\eval{\varphi} \in \promise\sharpPdottedcircle$ iff $\varphi$ is good.
\end{theorem}
\begin{proof}
    Theorem~\ref{thm:cat_main_theorem} gives the lower bound.

    For the upper bound, since $\varphi$ is a positive integer linear combination of finitely many basis functions and $\promise\sharpPdottedcircle$ is closed under such linear combinations it suffices to consider $\varphi(V) = \cntMsub{W}{V}$ for some vector space $W$.
    We thus need to count the number of $\dim(W)$-dimensional subspaces of $V \in \IF_p^d$, given $p^d$ and oracle access to $V$ as an input.
    For this now guess $p^{\dim(W)}$ many vectors in $\IF_p^d$ that form a $\dim(W)$-dimensional vector space (sorted in some way, such that the same vector space is only guessable in one way) and query the oracle whether all of them are part of $V$.
    If this is the case, accept, otherwise reject.
    Clearly this counts the number of $\dim(W)$-dimensional subspaces of $V$ and runs in time polynomial in $d$ ($\dim(W)$ is a fixed constant).
    We conclude $\eval{\varphi} \in \promise\sharpPdottedcircle$.
\end{proof}

Note that the upper bound in the previous algorithm is only a $\promise\sharpPdottedcircle$ upper bound.
The NTM $M$ is unable to verify whether the oracle indeed encodes a valid vector space.

\subsection{Parameter Sets}
\label{sec:cat:parsets}
Our final example are the so called parameter sets.
Fix some finite alphabet $A$.
An $n$-parameter set $X$ determined by a point $x^0 \in A^N$ and disjoint non-empty sets $\Lambda_1, \ldots, \Lambda_n \subseteq \{1, \ldots, N\}$ contains all points $f \in A^N$ with
\begin{align*}
    x_i = \begin{cases}
        x^0_i & \text{if $i \notin \bigcup_{j=1}^n\Lambda_j$},\\
        x_{i'} & \text{if $i, i' \in \Lambda_j$ for some $j = 1, \ldots, n$}.\\
    \end{cases}
\end{align*}
Note that different choices of $x^0$ or orderings of the $\Lambda_i$ can lead to the same parameter set and we do not distinguish between them, two parameter sets are the same if the underlying set is the same.
A $p$-parameter subset $Y \subseteq X$ is simply a $p$-parameter set $Y$ with $Y \subseteq X$ as sets. One can think of parameter sets as special affine vector spaces. 

We choose $C = P$ to be the same category $\parsets_{A}$.
The category $\parsets_{A}$ is the category of all parameter sets over $A$.
A morphism $f: X \to Y$ between parameter sets $X \subseteq A^N$ and $Y \subseteq A^{N'}$, is now a function between the underlying sets of $X$ and $Y$, such that for any parameter set $X' \subseteq X$, we have that $f(X') \subseteq Y$ is a parameter set.
Such a morphism $f$ is mono if the underlying function is injective and epi if it is surjective.
Furthermore $f$ is an isomorphism if the underlying function is an isomorphism.
This is only possible if $X$ and $Y$ have the same number of parameters.
But also conversely, if $X$ and $Y$ have the same number of parameters, then there is an isomorphism between them.
Since the morphisms $f: X \to Y$ are always a subset of the functions $X \to Y$ (as finite sets), we conclude that $\parsets_{A}$ is locally small and even locally finite. 
The proper ($\cE$, $\cM$)-factorization system is now given by choosing $\cE$ to be the set of all epimorphisms and $\cM$ to be the set of all monomorphisms.
Any morphism $f: X \to Y$ then %
uniquely factors through $\im(f)$ in a natural way.
The $\cM$-subobjects of some parameter set $X$ are canonically represented by the parameter sets $Y \subseteq X$, implying that $\parsets_{A}$ is finitely $\cM$-well-powered.

Furthermore the category has finite coproducts:
Remember that there is essentially (i.e.\ up to isomorphisms) only one parameter set for each number of parameters.
The coproduct of ``the'' $n$-parameter set $A^n$ and ``the'' $m$-parameter set $A^m$ is then ``the'' $n+m$-parameter set $A^{n+m}$.
The corresponding coproduct injections $f: A^n \to A^{n+m}$ and $g: A^m \to A^{n+m}$ simply map their inputs to either the first $n$ or the last $m$ parameters of $A^{n+m}$, respectively.
Both $f$ and $g$ are mono, implying the joint $\cM$-embedding property.

Using the canonical $n$-parameter set $A^n$, we see that the number of $m$-parameter sets $Y \subseteq A^n$ is given as the number of ways to partition $|A| \cup \{1, \ldots, n\}$ into $|A|+m$ disjoint non-empty subsets, such that the elements from $A$ are mapped into distinct subsets, w.l.o.g.\ the first $|A|$ ones.
These first $|A|$ subsets represent the parameters replaced by individual constants from $A$.
Thus, for an $m$-parameter set $X$ and an $n$-parameter set $Y$ we have
\[
    \cntMsub{X}{Y} = \left\{ \begin{array}{c} |A|+n \\ |A|+m \end{array} \right\}_{|A|}\,,
\]
a $|A|$-Stirling number of the second kind (see \cite{broder1984r} for details on $|A|$-Stirling numbers).
Similarly to the category $\finvec_{\IF_p}$, the motif parameters $\parsets_{A} \to \IQ$ simply depend on the number of parameters of the input.
The Graham-Rothschild Theorem (see \cite[Theorem 10.4]{Nesetril1995} for proof) states that $\parsets_{A}$ is $\cm$-Ramsey.
The blowup property for $C$ and $P$ is vacuously true for $P = C$.

The way we attach parameter sets as oracles to an NTM is as follows:
The NTM is given the size $|A|^N$ of an ambient space $A^N$ and some parameter set $X \subseteq A^N$.
It can query $X$ by asking the oracle, whether $(x_1, \ldots, x_N) \in X$ for some $x_1, \ldots, x_N \in A$.
The tuple $(x_1, \ldots, x_N)$ is suitably encoded as a binary string of some fixed length $\bitsize(N)$, only depending on $N$ and $A$, and thus $X$ itself can be encoded as a subset of $\{0, 1\}^{\bitsize(N)}$.
Our sets $C_{|A|^N}$ are then precisely all the parameter sets $X \subseteq A^N$ and $\bitsize(N)$ is poly-logarithmic in $|A|^N$.

The only thing left to apply Theorem~\ref{thm:cat_main_theorem} is to prove the existence of set-instantiators:
\begin{lemma}
    Let $\varphi: \parsets_{A} \to \IN$ be a motif parameter, let $M$ be a polynomial-time NTM computing $\eval{\varphi}$ and let $X \in \parsets_{A}$.
    Then there is a set-instantiator against $(M,X,\varphi)$.
\end{lemma}
\begin{proof}
    Let $X \in \parsets_{A}$ be an $n$-parameter set with sets $\Lambda_1, \ldots, \Lambda_n$.
    
    For now let the dimension $N$ of the ambient space be undetermined, we will choose it accordingly later.
    Fix some $z'^0 \in A^N$.
    Uniformly at random choose non-empty sets $\Lambda_1' \subseteq \{1, \ldots, \frac{N}{n}\}, \ldots, \Lambda_n' \subseteq \{1+\frac{n-N}{n}, \ldots, N\}$.
    This defines an $n$-parameter set $Z \subseteq A^N$.
    This directly defines a bijection $h$ from parameter sets $Y \subseteq X$ to parameter sets $Y' \subseteq Z$: Whenever a parameter from $X$ is replaced by a constant to obtain $Y$, we replace the same parameter in $Z$ by the same constant and whenever some parameters of $X$ are merged to obtain $Y$, we merge the same parameters of $Z$ to obtain $Y'$.
    This naturally also extends to a bijection between elements $x \in X$ and $z \in Z$ (they are $0$-parameter sets), so we can 
    Note that $Z = h(X)$.

    As usual we want to use the union bound to show that there is a choice of the $\Lambda_i'$ (and thus $h$), such that for every parameter set $Y \subseteq X$, all accepting paths of the computation $\#\acc_{M^{h(Y)}}(|A|^N)$ do not query the oracle for any elements in $h(Z) \setminus h(Y)$.

    If this happens with high probability with growing $N$ for all of the finitely many parameter sets $Y \subseteq X$, then eventually such an $h$ exists.
    For now fix some parameter set $Y \subseteq X$.
    Further group the possible $h$ by the image of $Y$.
    Call one such group $H$.
    Then the set $S_H$ of all elements queried by all accepting paths of the computation $\#\acc_{M^{h(Y)}}(|A|^N)$ is the same for all $h \in H$.
    Let $x \in h(X) \setminus h(Y)$ be fixed.
    Split $x$ into the natural blocks of $\frac{N}{n}$ coordinates.
    Then there is at least one block, say the $i$-th block, where $x$ needs to be changed to be in $h(Y)$.
    In particular, only one choice of entries in that block works.
    However, since $\Lambda_i'$ is uniformly distributed, those entries are different for every $\Lambda_i'$.
    We conclude
    \begin{align*}
        \Pr_{h \in H}[x \in S_H] \leq \frac{|S_H|}{2^{\frac{N}{n}}}\,.
    \end{align*}
    Combining this using the union bound and $|h(X) \setminus h(Y)| \leq |h(X)| = |X|$ we get
    \begin{align*}
        \Pr_{h \in H}[\exists x \in h(X) \setminus h(Y) \text{ with } x \in S_H] \leq |X| \cdot \frac{|S_H|}{2^{\frac{N}{n}}}\,.
    \end{align*}
    It remains to estimate $|S_H|$.
    Let $t_M(|A|^N)$ be the worst-case running time of $M$ on input $|A|^N$.
    Each oracle query can query exactly one element, there are at most $\varphi(h(X)) = \varphi(X)$ many accepting paths and each accepting path can query the oracle at most $t_M(|A|^N)$ times.
    This gives the bound of $|S_H| \leq \varphi(X) \cdot t_M(|A|^N)$.
    This bound is polynomial in $N$: $M$ is polynomial-time and $\varphi(X)$ is constant, due to $X$ being fixed.
    Across all possible functions $h$ (not just those from $H$), we thus get that all accepting paths of the computation $\#\acc_{M^{h(Y)}}(|A|^N)$ do not query the oracle for elements in $h(X) \setminus h(Y)$ with a probability of at least $1 - |X| \cdot \frac{\poly(N)}{2^{\frac{N}{n}}}$.
    
    Since $|X|$ and $n$ are constant, a final union bound proves that there is some large enough $N$, such that there is some $h$ that has our required property.

    We now construct our set-instantiator with $j=|A|^N$ for the $N$ above.
    This instantiation function for some parameter set $Y \subseteq X$ is $\inst_{\SI}(Y) = h(Y)$.
    For any computation path $\tau \in \{0, 1\}^\star$ of a computation $\#\acc_{M^{h(Y)}}(|A|^N)$ denote by $S_\tau \subseteq h(X)$ the elements that are queried by the computation.
    We set
    \begin{align*}
        \perc_{\SI}(\tau) := \begin{cases}
            \Span(h^{-1}(S_\tau)) & \text{if $\tau$ is an accepting path of the computation $\#\acc_{M^{h(X)}}(j)$,}\\
            \top & \text{otherwise}
        \end{cases}
    \end{align*}
    where $\Span(U)$ denotes the parameter set with the lowest amount of parameters containing $U$.
    We check the properties of a set-instantiator.
    Clearly $\varphi(\inst_{\SI}(Y)) = \varphi(Y)$ for all parameter sets $Y \subseteq X$ due to $Y$ and $h(Y)$ having the same number of parameters and thus being isomorphic.
    Further we need that for all parameter sets $Y \subseteq X$ we have that $\tau \in \{0, 1\}^\star$ is an accepting path for the computation $\#\acc_{M^{h(Y)}}(j)$ iff $\perc_\SI(\tau) \subseteq Y$.
    For this, fix some parameter set $Y \subseteq X$ and let $\tau$ be an accepting computation of the computation $\#\acc_{M^{h(Y)}}(j)$.
    By our choice of $h$ we know that no other elements of $h(X) \setminus h(Y)$ are queried, or in other words, $h^{-1}(S_\tau) \subseteq Y$ and thus $\tau$ is also an accepting path for the computation $\#\acc_{M^{h(X)}}(j)$.
    This directly implies $\perc_{\SI}(\tau) = \Span(h^{-1}(S_\tau)) \subseteq Y$.
    On the other hand, if $\perc_{\SI}(\tau) \subseteq Y$, then $\perc_{\SI} \neq \top$ and $\tau$ is an accepting path of the computation $\#\acc_{M^{h(X)}}(j)$.
    We conclude that $\tau$ is also an accepting path of the computation $\#\acc_{M^{h(Y)}}(j)$ due to $\perc_{\SI}(\tau) = \Span(h^{-1}(S_\tau)) \subseteq Y$.
\end{proof}

Applying Theorem~\ref{thm:cat_main_theorem} now proves the main result for $\parsets_{A}$.
\begin{theorem} \label{thm:cat:main:parsets}
    Let $\varphi: \parsets_{A} \to \IN$ be a motif parameter.
    Then $\eval{\varphi} \in \promise\sharpPdottedcircle$ iff $\varphi$ is good.
\end{theorem}
\begin{proof}
    Theorem~\ref{thm:cat_main_theorem} gives the lower bound.
    
    For the upper bound, since $\varphi$ is a positive integer linear combination of finitely many basis functions and $\promise\sharpPdottedcircle$ is closed under such linear combinations it suffices to consider $\varphi(X) = \cntMsub{Y}{X}$ for some $n$-parameter set $Y$.
    We thus need to count the number of $n$-parameter sets of $Y \subseteq X \subseteq A^N$, given $|A|^N$ and oracle access to $X$ as an input.
    For this now guess $|A|^n$ many elements in $A^N$ that form a $n$-parameter set (sorted in some way, such that the same parameter set is only guessable in one way) and query the oracle whether all of them are part of $X$.
    If this is the case, accept, otherwise reject.
    Clearly this counts the number of $n$-parameter subsets of $X$ and runs in time polynomial in $N$ ($n$ is a fixed constant).
    We conclude $\eval{\varphi} \in \promise\sharpPdottedcircle$.
\end{proof}

Again, the upper bound in the previous algorithm is only a $\promise\sharpPdottedcircle$ upper bound.
The NTM $M$ is unable to verify whether the oracle indeed encodes a valid parameter set.

\section{Non-negativity of the Example}
\label{sec:nonnegexample}
Recall from Example~\ref{ex:nonneg} the graph motif parameter
\[
f(G) = \indsub{\,\Gline}{G} - \indsub{\Gtria}{G} + \indsub{\GfourI}{G} + 2\,\indsub{\GfourII}{G} + 4\,\indsub{\GfourIII}{G}.
\]
We claim that $f(G)\geq 0$ for every $G$.
It is easy to see that $f(G)=\sum_{C} f(C)$, where the sum is over all connected components $C$ of $G$.
It remains to prove that $f(C)\geq 0$ for each connected component~$C$.
This is obvious for connected components with at most 2 vertices.
For every connected component $C$ that is a triangle, we have $f(C)=3-1=2 \geq 0$.
For every other connected component, for each triangle $T$ we choose an induced subgraph $F(T)$ of~$C$ that contains $T$ and that is isomorphic to $\GfourI$, $\GfourII$, or $\GfourIII$.
These (not necessarily unique) choices partition the set of triangles of $C$ into three disjoint sets $T(\GfourI)$, $\mathcal T(\GfourII)$ and $\mathcal T(\GfourIII)$ with $\indsub{\Gtria}{C} = |\mathcal T(\GfourI)| + |\mathcal T(\GfourII)| + |\mathcal T(\GfourIII)|$.

For an induced subgraph $H$ of $C$,
if $H\cong\GfourI$, then there is at most 1 triangle $T$ with $F(T)=H$.
If $H\cong\GfourII$, then there are at most 2 triangles $T$ with $F(T)=H$.
If $H\cong\GfourIII$, then there are at most 4 triangles $T$ with $F(T)=H$.
Hence,
\begin{eqnarray*}
f(C)
&\geq&
\indsub{\GfourI}{C} + 2\,\indsub{\GfourII}{C} + 4\,\indsub{\GfourIII}{C}-\indsub{\Gtria}{C}
\\
&=&
\big(\indsub{\GfourI}{C} - |\mathcal T(\GfourI)|\big)
 +\big(2\,\indsub{\GfourII}{C}-|\mathcal T(\GfourII)|\big) \\
& & \quad +\big(4\,\indsub{\GfourIII}{C}-|\mathcal T(\GfourIII)|\big)
\\
&\geq& 0.
\end{eqnarray*}

\end{document}

%% file: macros.tex
\newcommand{\defn}[1]{\emph{#1}}
\newcommand{\textsu}[1]{\textup{\textsf{#1}}}

\newcommand{\SI}{\textup{{\ensuremath{\mathcal S\!I}}}}
\newcommand{\inst}{\textup{\textsu{inst}}}
\newcommand{\perc}{\textup{\textsu{perc}}}

\newcommand{\im}{\textup{im}}

\newcommand{\convert}{\textup{\textsu{conv}}}

\newcommand{\ComCla}[1]{\textup{\textsu{#1}}}

\newcommand{\sharpP}{\ComCla{\#P}}

\newcommand{\GapP}{\ComCla{GapP}}

\newcommand{\PPAD}{\ComCla{PPAD}}

\newcommand{\PLS}{\ComCla{PLS}}

\newcommand{\TFNP}{\ComCla{TFNP}}

\newcommand{\acc}{\textup{\textsu{acc}}}

\newcommand{\IF}{\mathbb{F}}
\newcommand{\IQ}{\mathbb{Q}}

\newcommand{\IN}{\mathbb{N}}

\newcommand{\IZ}{\mathbb{Z}}

\newcommand{\supp}{\textup{supp}}

\newcommand{\starpad}{\, \star \,}

\def\bu{\bullet}

\def\nn{\mathbb N}

\def\qqq{\mathbb Q}

\def\cE{\mathcal E}

\def\cm{\mathcal M}
\def\cM{\mathcal M}

\def\<{\langle}
\def\>{\rangle}

\def\0{{\mathbf 0}}

\def\tilde{\widetilde}

\def\graphs{\mathcal{G}}
\def\graphsnoniso{{\mathcal{G}^\text{pure}}}

\def\ordgraphs{\mathcal{OG}}
\def\ordgraphsnoniso{{\mathcal{OG}^\text{pure}}}
\newcommand{\ordgraphssize}[1]{{\mathcal{OG}_{= #1}}}

\newcommand{\indsub}[2]{\ensuremath{\#\mathrm{Ind}(#1 \mathop{\to} #2)}} %
\newcommand{\sub}[2]{\ensuremath{\#\mathrm{Sub}(#1 \mathop{\to} #2)}}

\newcommand{\hominj}[2]{\ensuremath{\#\mathrm{InjHom}(#1 \mathop{\to} #2)}}
\newcommand{\strhominj}[2]{\ensuremath{\#\mathrm{StrInjHom}(#1 \mathop{\to} #2)}}
\renewcommand{\hom}[2]{\ensuremath{\#\mathrm{Hom}(#1 \mathop{\to} #2)}}
\newcommand{\aut}[1]{\#\mathrm{Aut}(#1)}

\newcommand{\ind}[2]{\#\mathrm{Ind}(#1 \to #2)}

\newcommand{\poset}[1]{\mathcal{P}(#1)}

\newcommand{\isoclass}[1]{[#1]}

\newcommand{\catsize}[2]{{{#1}_{=#2}}}
\newcommand{\Msub}[2]{\mathrm{Sub}_{\mathcal{M}}(#1 \mathop{\to} #2)}
\newcommand{\cntMsub}[2]{\#\mathrm{Sub}_{\mathcal{M}}(#1 \mathop{\to} #2)}
\newcommand{\domain}{\mathop{\mathrm{dom}}}
\newcommand{\codomain}{\mathop{\mathrm{cod}}}
\newcommand{\identity}{\mathrm{id}}
\newcommand{\poly}{\mathrm{poly}}
\newcommand{\polylog}{\mathrm{polylog}}
\newcommand{\bitsize}{\mathrm{bitsize}}
\newcommand{\finvec}{\mathrm{\mathbf{FinVect}}}
\newcommand{\parsets}{\mathrm{\mathbf{ParSets}}}
\newcommand{\finordgraphs}{\mathrm{\mathbf{FinOrdGraphs}}}
\newcommand{\Span}{\mathrm{Span}}

\newcommand{\rel}[1]{\mathrm{Rel}(#1)}
\newcommand{\mrel}[1]{\mathrm{MRel}(#1)}
\newcommand{\ordrel}[1]{\mathrm{ORel}(#1)}

\newcommand{\mordrel}[1]{\mathrm{MORel}(#1)}

\newcommand{\colgraphs}[1]{\mathrm{Col}\mathcal{G}(#1)}
\newcommand{\ppure}[1]{{#1}^{\text{$P$-pure}}}
\newcommand{\forb}[1]{\mathrm{Forb}(#1)}
\newcommand{\promise}{\mathsf{Pr}\text{-}}

\newcommand{\Gline}{%
\vcenter{\hbox{\begin{tikzpicture}%
\fill[black] (120:0.2) circle (0.15em);%
\fill[black] (240:0.2) circle (0.15em);%
\draw[thick] (120:0.2) -- (240:0.2);
\end{tikzpicture}}}%
}%

\newcommand{\Gtria}{%
\vcenter{\hbox{\begin{tikzpicture}%
\fill[black] (0:0.2) circle (0.15em);%
\fill[black] (120:0.2) circle (0.15em);%
\fill[black] (240:0.2) circle (0.15em);%
\draw[thick] (0:0.2) -- (120:0.2);
\draw[thick] (120:0.2) -- (240:0.2);
\draw[thick] (240:0.2) -- (0:0.2);
\end{tikzpicture}}}%
}%

\newcommand{\Gpath}{%
\vcenter{\hbox{\begin{tikzpicture}%
\fill[black] (0:0.2) circle (0.15em);%
\fill[black] (0:0.4) circle (0.15em);%
\fill[black] (0:0.6) circle (0.15em);%
\draw[thick] (0:0.2) -- (0:0.4);
\draw[thick] (0:0.4) -- (0:0.6);
\end{tikzpicture}}}%
}%

\newcommand{\GfourI}{%
\vcenter{\hbox{\begin{tikzpicture}%
\fill[black] (0:0.4) circle (0.15em);%
\fill[black] (0:0.2) circle (0.15em);%
\fill[black] (120:0.2) circle (0.15em);%
\fill[black] (240:0.2) circle (0.15em);%
\draw[thick] (0:0.2) -- (0:0.4);
\draw[thick] (0:0.2) -- (120:0.2);
\draw[thick] (120:0.2) -- (240:0.2);
\draw[thick] (240:0.2) -- (0:0.2);
\end{tikzpicture}}}%
}%

\newcommand{\GfourII}{%
\vcenter{\hbox{\begin{tikzpicture}%
\fill[black] (0:0.4) circle (0.15em);%
\fill[black] (0:0.2) circle (0.15em);%
\fill[black] (120:0.2) circle (0.15em);%
\fill[black] (240:0.2) circle (0.15em);%
\draw[thick] (0:0.4) to[bend right] (120:0.2);
\draw[thick] (0:0.4) to[bend left] (240:0.2);
\draw[thick] (0:0.2) -- (120:0.2);
\draw[thick] (120:0.2) -- (240:0.2);
\draw[thick] (240:0.2) -- (0:0.2);
\end{tikzpicture}}}%
}%

\newcommand{\GfourIII}{%
\vcenter{\hbox{\begin{tikzpicture}%
\fill[black] (0:0.4) circle (0.15em);%
\fill[black] (0:0.2) circle (0.15em);%
\fill[black] (120:0.2) circle (0.15em);%
\fill[black] (240:0.2) circle (0.15em);%
\draw[thick] (0:0.4) to[bend right] (120:0.2);
\draw[thick] (0:0.4) to[bend left] (240:0.2);
\draw[thick] (0:0.2) -- (0:0.4);
\draw[thick] (0:0.2) -- (120:0.2);
\draw[thick] (120:0.2) -- (240:0.2);
\draw[thick] (240:0.2) -- (0:0.2);
\end{tikzpicture}}}%
}%

\newcommand{\sharpPdottedcircle}{{\sharpP}^{\protect\begin{tikzpicture}
\draw[line width=0.25ex, dash pattern=on 1.2pt off 1.2pt] (0,0) circle (0.55ex);
\end{tikzpicture}}}

\newcommand{\eval}[1]{\operatorname{Eval}_{#1}}